\documentclass[12pt,letterpaper]{article}
\usepackage{hyperref}
\usepackage{amssymb,amsmath}
\usepackage{mathtools}
\usepackage{amsthm}
\usepackage{bbm}
\usepackage{graphicx,epsf}
\usepackage{epstopdf}
\usepackage{enumerate}
\usepackage{caption}
\usepackage{natbib}
\usepackage{url}
\usepackage{bm}
\usepackage{placeins}
\DeclareMathOperator{\median}{\mbox{median}}

\DeclareMathOperator{\argmin}{\mbox{argmin}}
\DeclareMathOperator{\MD}{\mbox{MD}}
\DeclareMathOperator{\RD}{\mbox{RD}}
\DeclareMathOperator{\KL}{\mbox{KL}}
\DeclareMathOperator{\tr}{\mbox{tr}}

\newcommand{\eps}{\varepsilon}
\newcommand{\gs}{\geqslant}
\newcommand{\ls}{\leqslant}
\newcommand{\bzero}{\boldsymbol 0}
\newcommand{\bdot}{\boldsymbol{\cdot}}

\newcommand{\bu}{\boldsymbol u}
\newcommand{\bv}{\boldsymbol v}

\newcommand{\bz}{\boldsymbol z}
\newcommand{\btz}{\boldsymbol{\tilde{z}}}
\newcommand{\bA}{\boldsymbol A}
\newcommand{\bB}{\boldsymbol B}
\newcommand{\bC}{\boldsymbol C}
\newcommand{\bI}{\boldsymbol I}
\newcommand{\bW}{\boldsymbol W}
\newcommand{\bX}{\boldsymbol X}
\newcommand{\bY}{\boldsymbol Y}
\newcommand{\bZ}{\boldsymbol Z}
\newcommand{\bsX}{\boldsymbol{\dot{X}}}
\newcommand{\btX}{\boldsymbol{\tilde{X}}}
\newcommand{\btY}{\boldsymbol{\tilde{Y}}}
\newcommand{\bbeta}{\boldsymbol \beta}
\newcommand{\bdelta}{\boldsymbol \delta}
\newcommand{\bnabla}{\boldsymbol \nabla}
\newcommand{\btheta}{\boldsymbol \theta}
\newcommand{\bhbeta}{\boldsymbol{\hat{\beta}}}
\newcommand{\bhdelta}{\boldsymbol{\hat{\delta}}}
\newcommand{\bhmu}{\boldsymbol{\hat{\mu}}}
\newcommand{\bhtheta}{\boldsymbol{\hat{\theta}}}
\newcommand{\bmu}{\boldsymbol \mu}
\newcommand{\bSigma}{\boldsymbol \Sigma}
\newcommand{\bhSigma}{\boldsymbol{\widehat{\Sigma}}}
\newcommand{\btSigma}{\boldsymbol{\tilde{\Sigma}}}
\newcommand{\tih}{\tilde{h}}

\newtheorem{proposition}{Proposition}

\begin{document}

\def\spacingset#1{\renewcommand{\baselinestretch}%
{#1}\small\normalsize} \spacingset{1}


\title{\bf Handling cellwise outliers by sparse 
	     regression and robust covariance}
\author{Jakob Raymaekers and Peter J. 
	  Rousseeuw \hspace{.1cm} \\ \\
		Section of Statistics and Data Science,\\
		Department of Mathematics, KU Leuven, Belgium}
\date{December 7, 2020}
\maketitle

\begin{abstract}
We propose a data-analytic
method for detecting cellwise outliers. 
Given a robust covariance matrix, outlying 
cells (entries) in a row are found by the cellHandler 
technique which combines
lasso regression with a stepwise application of 
constructed cutoff values. 
The penalty term of the lasso has a physical 
interpretation as the total distance 
that suspicious cells need to move in order 
to bring their row into the fold. 
For estimating a cellwise robust covariance matrix 
we construct a detection-imputation method which 
alternates between flagging outlying cells and 
updating the covariance matrix as in the EM algorithm.  
The proposed methods are illustrated 
by simulations and on real data about volatile organic
compounds in children.
\end{abstract}


\spacingset{1.45} 


\section{Introduction}
It is a fact of life that most real data sets 
contain outliers, that is, elements that do not
fit in with the majority of the data.
These outliers can be annoying errors, but may
also contain valuable information.
In either case, finding them is of 
practical importance.
In statistics this is called outlier detection, 
and in the computer science literature it is
also called anomaly detection or exception 
mining, see e.g. \cite{Chandola2009}.

The most common paradigm is that of casewise 
outliers, which assumes that most cases were
drawn from a certain model distribution but
some other cases were not.
The latter are also called rowwise outliers, since
data often comes in the form of a table (matrix)
in which the rows are the cases and the columns
represent the variables.
In computer science one often uses outlier 
detection methods based on Euclidean distances, 
which by construction are invariant for
orthogonal transformations of the rows.
In statistics many outlier detection methods
are also invariant for affine transformations, 
i.e. nonsingular linear transforms
combined with shifts. 

The study of cellwise outliers is a more
recent research topic. 
This is the situation where some individual
cells (entries) of the data matrix deviate from 
what they should have been.
\cite{alqallaf2009} first formulated this
paradigm.
Note that cells are intimately tied to the
coordinate system, whereas orthogonal or other
linear transformations would change the cells. 
To illustrate the difference between the
rowwise and cellwise approaches, consider
the standard multivariate Gaussian model
in dimension $d=4$ with
the suspicious point $(10,0,0,0)$.
By an orthogonal transformation of the data
this point can be moved to 
$(\sqrt{50},\sqrt{50},0,0)$ or to
$(5,5,5,5)$ so any orthogonally invariant 
rowwise detection method will treat all three 
situations the same way.
But in the cellwise paradigm $(10,0,0,0)$ 
has one outlying cell, 
$(\sqrt{50},\sqrt{50},0,0)$ has two and
$(5,5,5,5)$ has four.

For an illustration of cellwise outliers
see Figure \ref{fig:cellmap}. 
It depicts part of a dataset that will be
described later.
The rows are cases and the columns are
variables. 
The regular cells are shown in yellow.
Red colored cells indicate that their value 
is higher than expected, while blue cells 
indicate unusually low values.

\begin{figure}[!ht]
\center
\vskip0.2cm
\includegraphics[width = 0.55\textwidth]
  {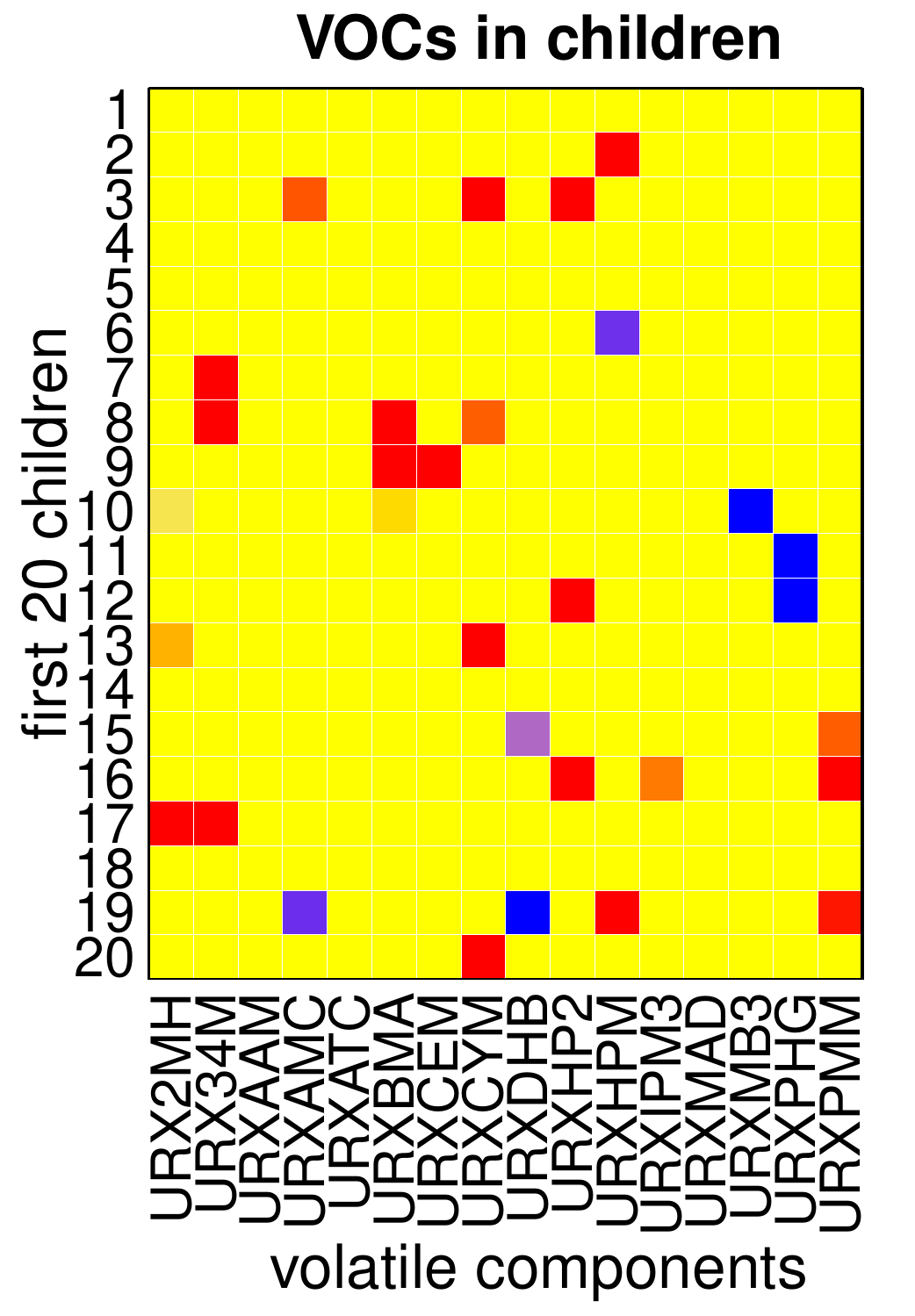}
\vskip-0.2cm
\caption{Illustration of cellwise outliers.
Red squares indicate cells with unexpectedly
high values, and blue squares indicate
unusually low values. Regular cells are
yellow.}
\label{fig:cellmap}
\end{figure}

When the model has substantially correlated
variables the cellwise outliers need not be
marginally outlying, and then it can be quite
hard to detect them.
\cite{VanAelst2011} proposed one of the first 
methods, based on an outlyingness measure of
the Stahel-Donoho type. 
\cite{Farcomeni2014} looks for the cells that,
when put to missing, yield the highest
Gaussian partial likelihood.
\cite{Agostinelli2015} and \cite{Leung2017} 
use a univariate or bivariate filter on the 
variables to flag cellwise outliers,
followed by S-estimation.
\cite{Rousseeuw2018} predict the values of all 
cells and flag the observed cells that differ 
much from their prediction.
\cite{Debruyne2019} consider rowwise outliers
and ask which variables contribute
the most to their outlyingness.
The O3 plot of \cite{Unwin:O3} visualizes
cases that are outlying in lower dimensions.

There has also been substantial work to estimate
the covariance matrix underlying the model in
the presence of cellwise outliers, which will
be be briefly reviewed in Section 
\ref{sec:existing}.

Most of the statistical research on cellwise
outliers has focused on the FICM contamination
model of \cite{alqallaf2009} which assumes that 
the outlying cells come from a single
distribution, and typically this distribution
has all its mass in a single value $\gamma$.
Here we will not restrict ourselves to that 
setting, and in the simulations we will allow
for the cellwise outlying values to depend on 
which cells are contaminated, creating 
structured cellwise outliers.
This is a more challenging problem, and it is 
clear that the underlying covariance structure 
will play a role.

Note that the multivariate setting is very 
different from regression with a univariate 
response.
In regression, having a robust fit is sufficient
for flagging outlying responses,\linebreak 
because their residuals from the robust fit will 
be large in absolute value. 
In the multivariate situation it is much harder:
even if we knew the true pre-contamination 
covariance matrix $\bSigma$, how would we find
the anomalous data cells? 
Currently no method is available to do this.
Our aim is to fill that gap by constructing 
such a method called cellHandler, 
described in Section \ref{sec:cellHandler}.
To estimate a cellwise robust covariance matrix
$\bhSigma$, Section \ref{sec:DI} constructs the 
detection-imputation algorithm which alternates 
between cellHandler and re-estimating $\bSigma$
as in the EM algorithm.
In Section \ref{sec:simulation} the performance 
of this approach is studied by simulation, and
Section \ref{sec:example} analyzes real data
on volatile organic compounds in children.

\section{The cellHandler method}
\label{sec:cellHandler}
In this section we construct a method to detect 
outlying cells when the true positive definite
covariance matrix $\bSigma$ is known.
In reality $\bSigma$ is usually unknown,
but this method is a major component of the
algorithm proposed in the next section for 
estimating $\bSigma$.
 
\subsection{Ranking cells by their outlyingness}
\label{sec:ranking}
We start by standardizing the columns (variables) 
of the dataset, using robust univariate estimates 
of location and scale such as the median and the 
median absolute deviation. 
This also ensures that the result will be 
equivariant to shifting and rescaling of the
original variables.
The resulting $d$-variate cases are denoted as 
$\bz_i$ for $i=1,\ldots,n$\,.

For a given case $\bz$, the central question in 
this section is how we can identify the cells 
that are most likely to be contaminated.
Any set of cells in $\bz$ may be contaminated, 
and while it may be tempting to somehow 
investigate all $2^d$ subsets of $\bz$ this 
quickly becomes infeasible due to the 
exponential complexity in $d$.
Therefore we need a different approach to provide
candidate cells that may be contaminated while 
avoiding an insurmountable computational cost. 
Note that the squared Mahalanobis distance 
$\MD^2(\bz,\bmu,\bSigma) = (\bz-\bmu)' \bSigma^{-1}
(\bz - \bmu)$ measures how far $\bz$ lies from 
the uncontaminated distribution.
The idea is to reduce the Mahalanobis distance of 
$\bz$ by changing only a few cells. 
Mathematically, we look for a  
$d$-variate vector $\bdelta$ such that 
$\MD^2(\bz - \bdelta, \bmu, \bSigma)$ is small. 
Interestingly, this problem can be rewritten in an 
elegant form, as presented in the following 
proposition.
\begin{proposition} \label{prop:1}
Modifying cells to reduce the Mahalanobis distance 
of their row can be rewritten using the sum of 
squares in a linear model.
\end{proposition}
\begin{proof}
Observe that 
\begin{align}\label{eq:MDreg}
\MD^2(\bz- \bdelta, \bmu, \bSigma) 
&=(\bz-\bdelta-\bmu)' \bSigma^{-1} (\bz-\bdelta-\bmu)
   \nonumber \\
&= ||\bSigma^{-1/2}(\bz-\bdelta-\bmu)||_2^2 \nonumber \\
&= ||\bSigma^{-1/2}(\bz-\bmu)-\bSigma^{-1/2}\bdelta||_2^2
   \nonumber \\
&= ||\btY - \btX\bdelta||_2^2 
\end{align}
which is the objective of a regression without
intercept with known response vector\linebreak 
$\btY \coloneqq \bSigma^{-1/2}(\bz-\bmu)$ and 
predictor matrix $\btX \coloneqq \bSigma^{-1/2}$ 
with coefficient vector $\bdelta$. Here 
$\bSigma^{-1/2}$ is the unique PD inverse root
of $\bSigma$.
\end{proof}
It is clear that the ordinary least squares (OLS) 
solution to \eqref{eq:MDreg} is 
$\bhdelta_{LS} = \bz-\bmu$ since it makes the sum 
of squares zero. 
However, using $\bhdelta_{LS}$ would replace
the entire row by the vector $\bmu$, which would
lose the information in the non-outlying cells.
We prefer to change as few cells as possible,
so we want a sparse coefficient vector 
$\bhdelta$.
A natural choice for this problem is the lasso
\citep{tibs:lasso},
given by the minimization of 
\begin{equation} \label{eq:lasso}
 ||\btY - \btX\bdelta||_2^2
  + \lambda ||\bdelta||_1
\end{equation}
where $||\bdelta||_1 = 
|\delta_1| + \ldots + |\delta_d|$\;.
Lasso regression penalizes $||\bdelta||_1$ which 
yields a path of sparse solutions to the 
regression problem for decreasing value of $\lambda$. 
Note that the penalty term $||\bdelta||_1$ has a concrete 
physical meaning in this setting: it is the total 
distance which the corresponding cells of $\bz$ need 
to travel in order to bring $\bz$ into the fold. 
This is unusual, as the $L^1$ term is typically included 
as a device to induce sparsity without a specific 
subject-matter interpretation. 

The description is not yet complete, because we
have to take special care of cells $z_j$ that lie 
far away. Moving such cells into place requires
large components $\delta_j$ which inflate 
the penalty term, so these $z_j$ would appear rather 
late in the lasso path.
Fortunately such far marginal outliers $z_j$ are easy 
to spot, as they have a large univariate outlyingness
$O_j = |z_j - \mu_j|/\sqrt{\Sigma_{jj}}$\,.
Therefore we downweight the $\delta_j$ in the
penalty term by a factor $w_j = \min(1,1.5/O_j)$
which is the weight associated with the univariate 
Huber M-estimator.
This replaces $||\bdelta||_1$ in \eqref{eq:lasso}
by $||\bW \bdelta||_1$ where 
$\bW \coloneqq diag(w_1,\ldots,w_d)$\,.
Note that this weighted lasso can be rewritten as 
a plain lasso as follows.
Since all the weights are strictly positive $\bW$
is invertible, so we can write 
$\btX\bdelta = (\btX \bW^{-1})(\bW \bdelta) =
\bsX \bbeta$ where $\bsX \coloneqq \btX \bW^{-1}$ 
and $\bbeta \coloneqq \bW \bdelta$. 
This merely changes the units of the variables
in $\btX$, so we minimize
\begin{equation} \label{eq:lasso2}
||\btY - \bsX\bbeta||_2^2  + \lambda ||\bbeta||_1
\end{equation}
followed by transforming $\bhbeta$ back to $\bhdelta$.
The penalty term $||\bbeta||_1$ keeps its 
interpretation in the new units determined by 
the $w_j$.

Note that lasso steps do not only add variables:
sometimes they take a variable out of the model.
But in our context it is natural to impose that once 
a cell is flagged it stays flagged, i.e. that a 
selected regressor stays in the model. 
By imposing this constraint we arrive at the elegant 
and fast least angle regression (LAR) algorithm of 
\cite{efron2004}. 
This is the option \texttt{type="lar"} in the 
R-package \texttt{lars} \citep{HastieLars}, and its
performance turned out to be very similar to that
of \texttt{type="lasso"} in our setting. 
Using LAR also simplifies and speeds up the next 
step of cellHandler in Section 2.2.

The way LAR works in our problem is intuitive.
The gradient of $\MD^2(\bz- \bdelta, \bmu, \bSigma)=
||\btY - \bsX\bbeta||_2^2$ with respect to $\bbeta$ 
is $\bnabla = -2\bsX'(\btY-\bsX\bbeta)$.
This gradient $\bnabla=(\nabla_1,\ldots,\nabla_d)$
is zero at the minimum of $\MD^2$, i.e. when
$\bbeta$ is the OLS fit 
$(\bsX'\bsX)^{-1}\bsX\btY = \bW(\bz-\bmu)$.
LAR first takes the coordinate with
highest $|\nabla_j|$ and moves $\beta_j$, 
i.e. cell $j$, to
reduce $|\nabla_j|$ until it equals the second
largest $|\nabla_h|$. Then it 
moves cells $j$ and $h$ such that
$|\nabla_j|=|\nabla_h|$ decrease together, until
it reaches the third largest $|\nabla_m|$, and so on.

For each row $\bz$ we have now obtained a ranking
of its cells, corresponding to the order in which 
they occurred in the path for reducing
$\MD^2(\bz-\bdelta,\bmu,\bSigma)$.
Each row $\bz$ can have its own $\bdelta$.

\subsection{Handling outlying cells}
\label{sec:handling}

After $k$ steps of LAR we have a set of $k$
candidate cells.
The question is whether these candidate cells
are sufficient.
In other words, is it possible to edit these 
$k$ cells while keeping the remaining $d-k$ cells
intact, in such a way that the edited row 
behaves like a clean row?
To this end we will edit the $k$ candidate cells 
to maximize the Gaussian likelihood given the 
remaining cells.
Suppose without loss of generality that the 
candidate cells are the first $k$ entries of $\bz$  
so we can denote
$\bz' = [\bz_1' \; \bz_2']\,$ 
and $\bmu' = [\bmu_1' \; \bmu_2']\,$ 
where $\bz_1'$ and $\bmu_1'$
have length $k$. Also write
$\bSigma_{11}$ for the upper left submatrix of 
$\bSigma$ of size $k \times k$ and so on. 
As in the E step of the EM algorithm 
(see e.g. \cite{Little:EM}),
maximizing the Gaussian likelihood implies that 
$\bz_1$ should be shifted to 
$E_{\bmu, \bSigma}[\bZ_1| \bZ_2 = \bz_2] = 
\bmu_1 + \bSigma_{12}\bSigma_{22}^{-1} 
(\bz_2 - \bmu_2)$.

This imputation appears to require inverting 
the submatrix $\bSigma_{22}$ of $\bSigma$. 
However, it can also be obtained by OLS 
regression in the model \eqref{eq:MDreg} of 
Proposition 1 but restricted to the set of
$k$ candidate variables. 
This is shown in the following proposition, the 
proof of which is given in Section \ref{A:Estep} 
of the Appendix.

\begin{proposition} \label{prop:2}
Let the $k$-variate $\bhtheta_1$ 
be the OLS fit to the regression problem 
given by 
\begin{equation*}
\argmin_{\btheta} ||\bSigma^{-1/2}(\bz-\bmu)-
    (\bSigma^{-1/2})_{\bdot 1} \btheta_1||_2^2 
\end{equation*}
where $(\bSigma^{-1/2})_{\bdot 1}$ denotes the first 
$k$ columns of the matrix $\bSigma^{-1/2}$. Then
\begin{equation*}
  \bz_1 - \bhtheta_1 =  \bmu_1 + 
  \bSigma_{12}\bSigma_{22}^{-1} (\bz_2 - \bmu_2)\;\;.
\end{equation*}
\end{proposition}

\noindent In the implementation of cellHandler 
these vectors
$\bhtheta_1$ are obtained as a byproduct of the LAR
algorithm without extra computational cost, see
Section \ref{A:cellHandler} of the Appendix.
This means that it carries out the above computation
for $k=1,\ldots,d$ without having to invert any
matrix.

We now have a sequence of length $d$ of cells
in $\bz$, with their possible imputations at 
every stage $k$.
The question remains where to stop in this path, 
i.e. how many cells should we actually flag? 
For that we use the following proposition:

\begin{proposition} \label{prop:3}
For every $1 \ls k \ls d$ we have:
\begin{enumerate}
\item The residual sum of squares 
  $\mbox{RSS}_k  = ||\bSigma^{-1/2}(\bz - \bmu)- 
	  (\bSigma^{-1/2})_{\bdot 1} \bhtheta_1||_2^2$
	of the OLS fit $\bhtheta_1$ to the first $k$ 
	cells in the path
	equals the squared partial Mahalanobis distance 
	$\MD^2(\bz_2,\bmu_2,\bSigma_{22}) =
 (\bz_2-\bmu_2)'\bSigma_{22}^{-1}(\bz_2-\bmu_2)\,$.
\item For Gaussian data the difference between 
  two subsequent $\mbox{RSS}$ follows the $\chi^2$ 
	distribution with 1 degree of freedom, i.e. 
	$\Delta_k \coloneqq \mbox{RSS}_{k-1} -
	 \mbox{RSS}_k \sim \chi^2(1)$.
\end{enumerate}
\end{proposition}

The proof is in Section \ref{A:chisq} of 
the Appendix.
The distributional assumption in part 2 is
unrealistic in our setting, but at least the
result provides a rough yardstick that we can
use in our data-analytic procedure.
Following the path for $1 \ls k \ls d$ we will
compare the $\Delta_k$\, to a cutoff
$q$, say the 0.99 quantile of $\chi^2(1)$,
and flag the cells with  $\Delta_k > q$.

We illustrate cellHandler by two simple bivariate 
examples. 
The left part of Figure \ref{fig:DOA} assumes that 
the true $\bmu = \bzero$ and that $\bSigma$ is the 
identity matrix, so the correlation $\rho$ is zero. 
For any point $\bz = [z_1\;\;z_2]'\,$ we can then 
run cellHandler to see which of these cells are 
flagged, if any. 
In the central square no cells are flagged, to its 
left and right $z_1$ is flagged, above and below it 
$z_2$ is flagged, and in the outer regions both 
$z_1$ and $z_2$ are flagged. 
Things get more eventful when $\bSigma$ has 1 on 
the diagonal and $\rho = 0.9$ elsewhere. 
In the right panel of Figure \ref{fig:DOA} we see 
that no cells are flagged when $\bz$ lies in 
part of an elliptical region. 
The domain where only $z_1$ is flagged now has a 
more complicated form, and the same holds for 
$z_2$\,, whereas the region in which both are 
flagged is similar to before. 
Of course the main purpose of cellHandler is to 
deal with higher dimensions, which are harder to 
visualize. 

\begin{figure}[!ht]
\center
\vskip0.2cm
\includegraphics[width = 0.49\textwidth]
  {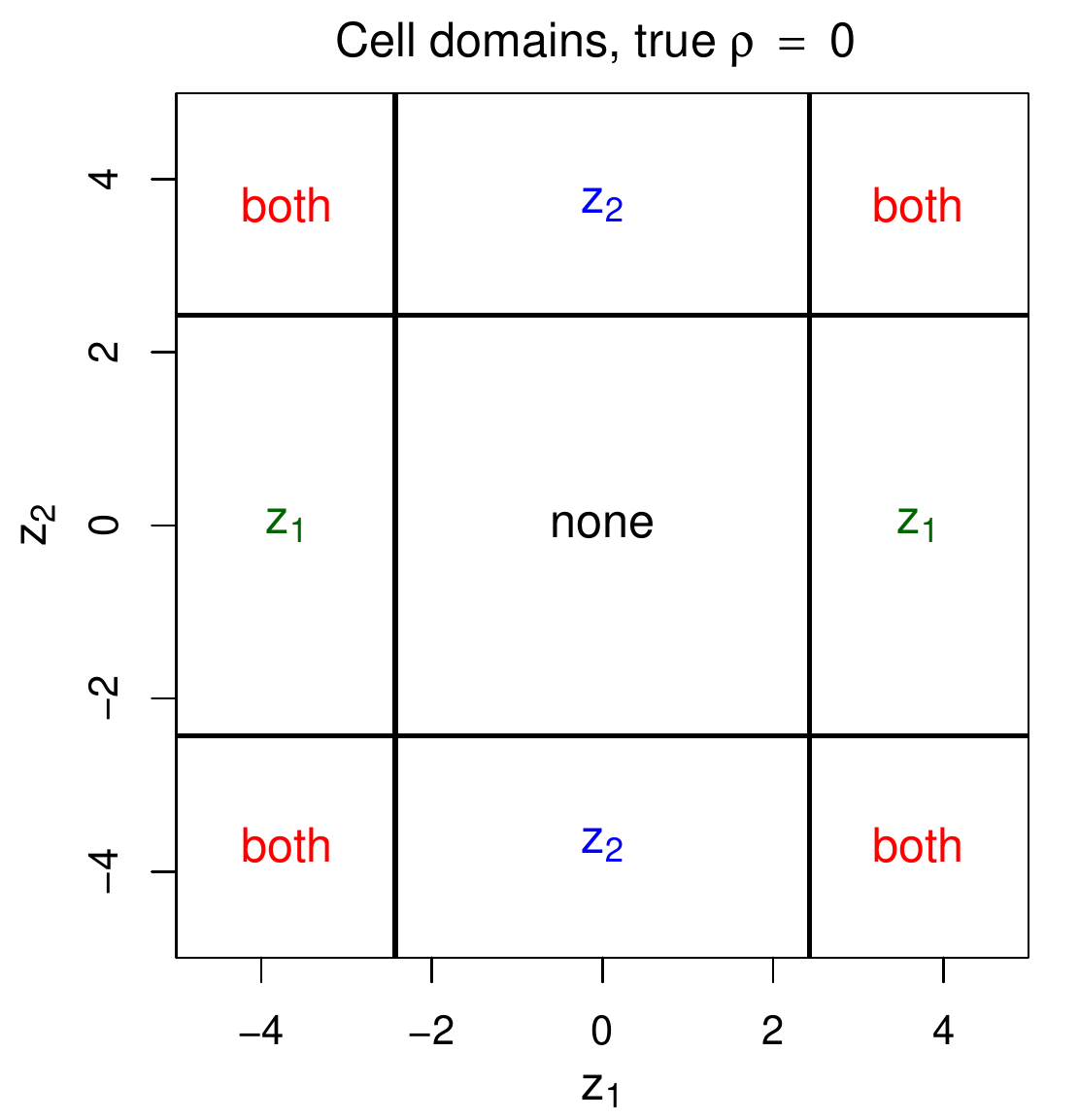}
\includegraphics[width = 0.49\textwidth]
  {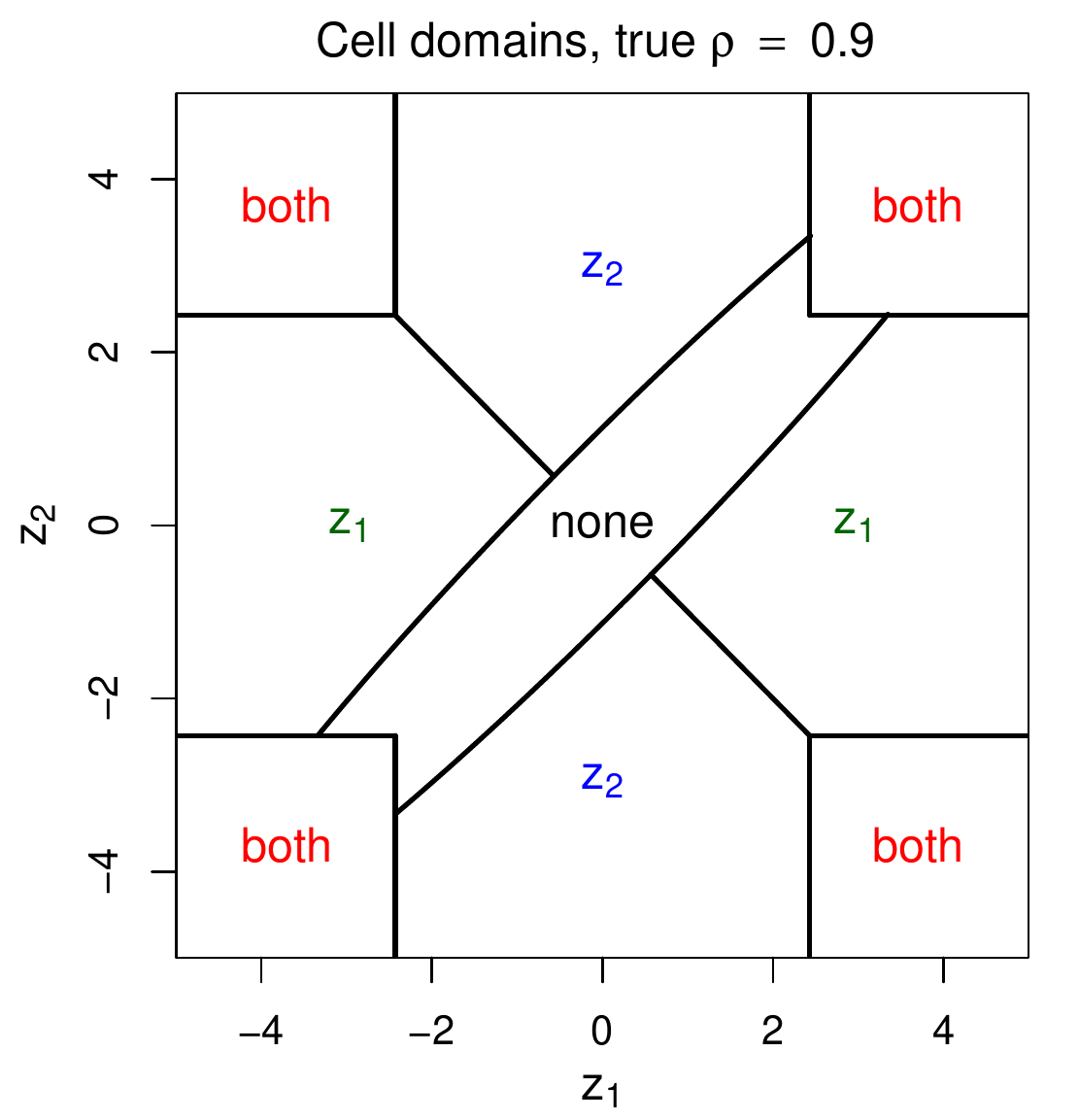}
\vskip-0.2cm
\caption{Bivariate domains where no cells are 
  flagged, where only $z_1$ is flagged, 
	where only $z_2$ is flagged, and where both are 
	flagged, when the true correlation is $\rho=0$ 
	(left panel) and when $\rho=0.9$ (right panel).}
\label{fig:DOA}
\end{figure}

\subsection{Simulation}
\label{subsec:cFsim}
To evaluate the performance of cellHandler 
we run a small simulation in which the 
uncontaminated data are $d$-variate Gaussian 
with $\bmu = \bzero$.
Since cellwise methods are neither affine 
or orthogonal invariant, we consider 
underlying covariance matrices $\bSigma$ 
of two types.  
Type ALYZ are the randomly 
generated correlation matrices of 
\cite{Agostinelli2015} which typically have 
relatively small correlations.
Type A09 is given by 
$\Sigma_{jh} := (-0.9)^{|j-h|}$ and contains 
both large and small correlations.

The outlying cells are generated as follows. 
The positions of the cells to be contaminated
are obtained by randomly drawing $n \eps$
indices in each column of the data matrix.
Then we look at each row $(z_1,\ldots,z_d)$ 
with such cells, and denote the indices of those 
cells as the set
$K = \{j(1),\ldots,j(k)\}$ of size $k$.
Next, we replace $(z_{j(1)},\ldots,z_{j(k)})$
by the $k$-dimensional row 
$\bv = \gamma \sqrt{k}\, \bu'/\MD(\bu,\bmu_K, 
 \bSigma_K)$
where $\bmu_K$ and $\bSigma_K$ are restricted
to the indices in $K$ and where $\bu$ is the
eigenvector of $\bSigma_K$ with smallest 
eigenvalue.
This procedures generates $\bv$ which are
structurally outlying in the subspace of the 
coordinates in $K$, while many of 
these cells will not be marginally outlying.
This produces cellwise outliers that are more 
challenging than in the earlier literature, 
which used $\bv = (\gamma,\ldots,\gamma)$.

\begin{figure}[!hb]
\center
\vskip-0.5cm
\includegraphics[width = 0.43\textwidth]
   {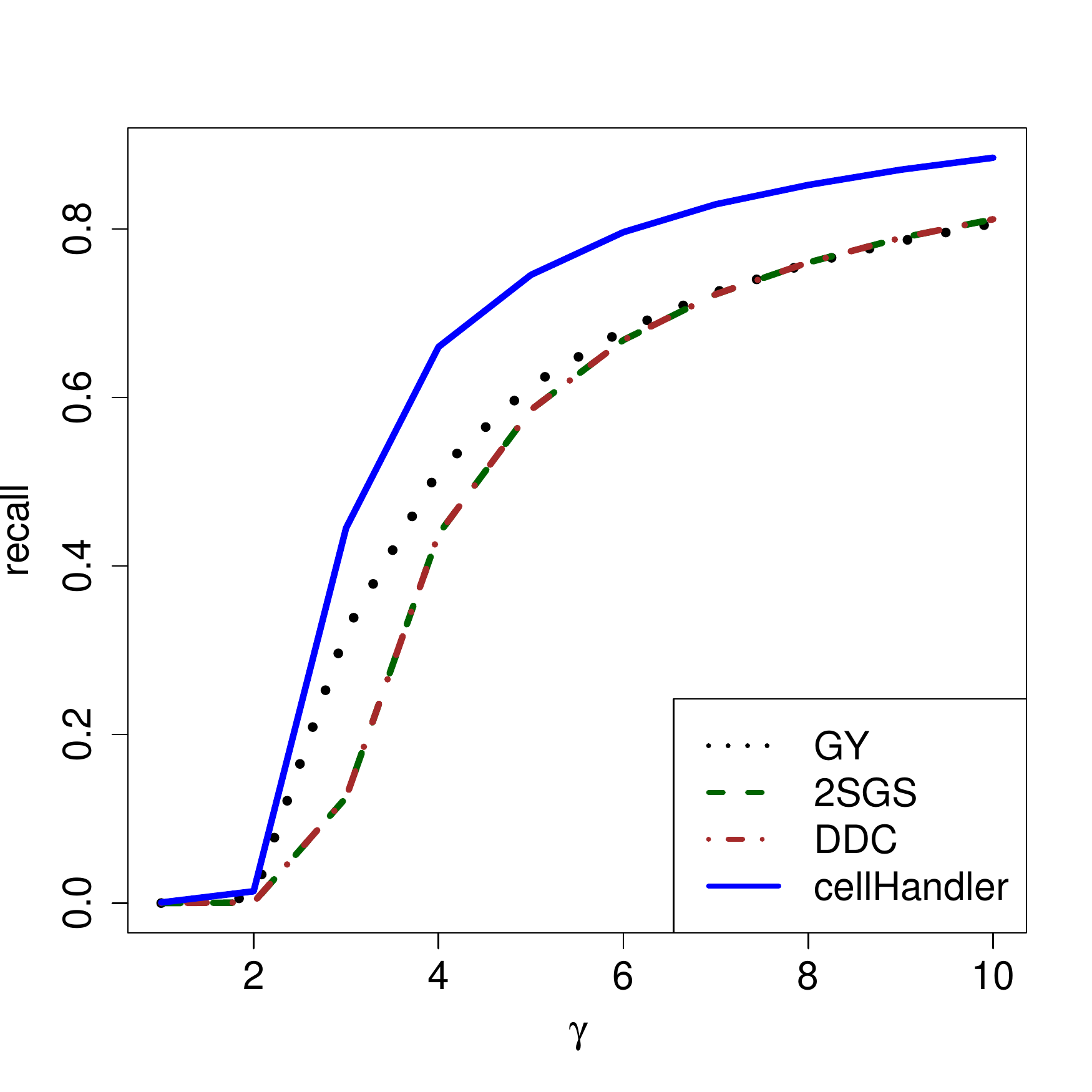}
\includegraphics[width = 0.43\textwidth]
   {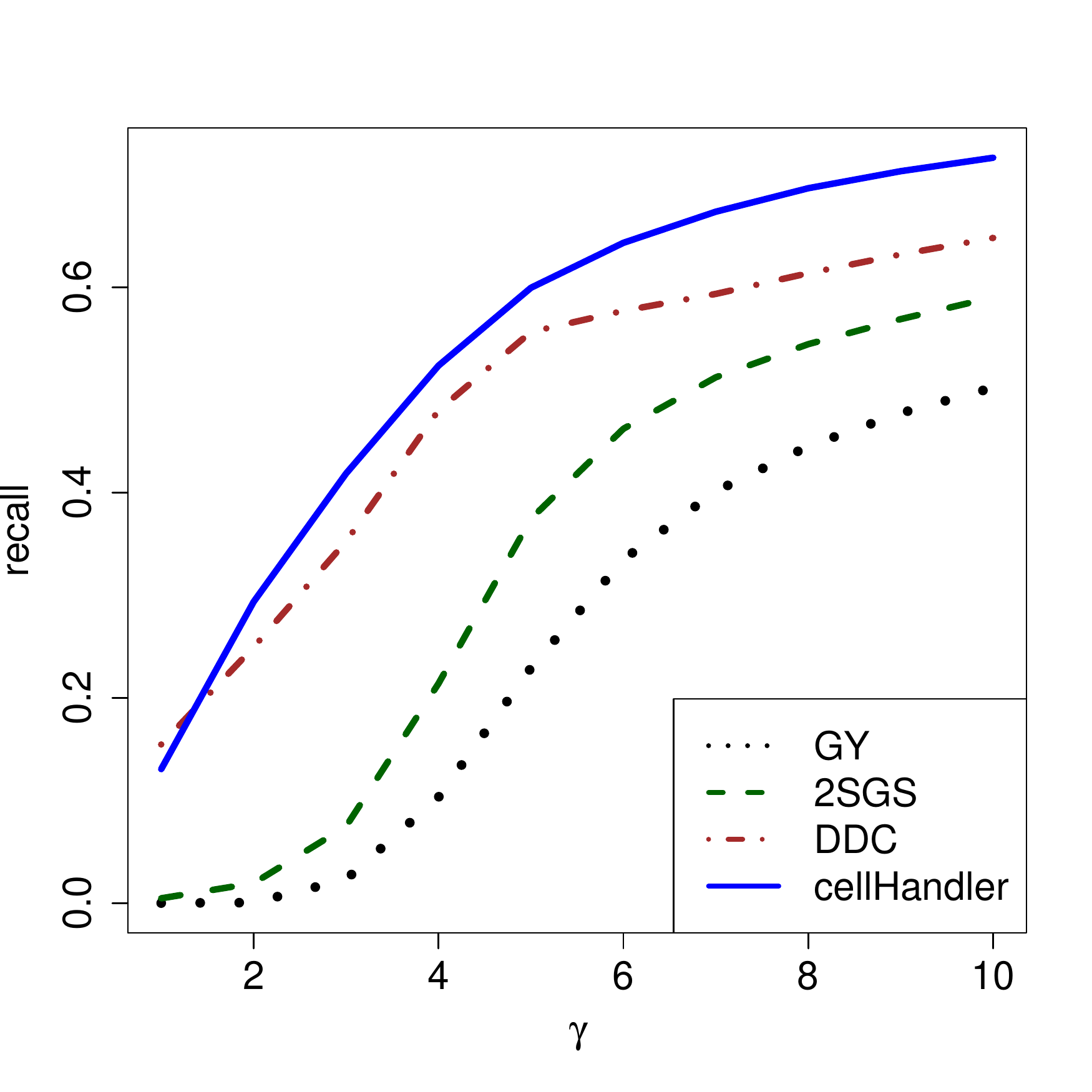}
\includegraphics[width = 0.43\textwidth]
   {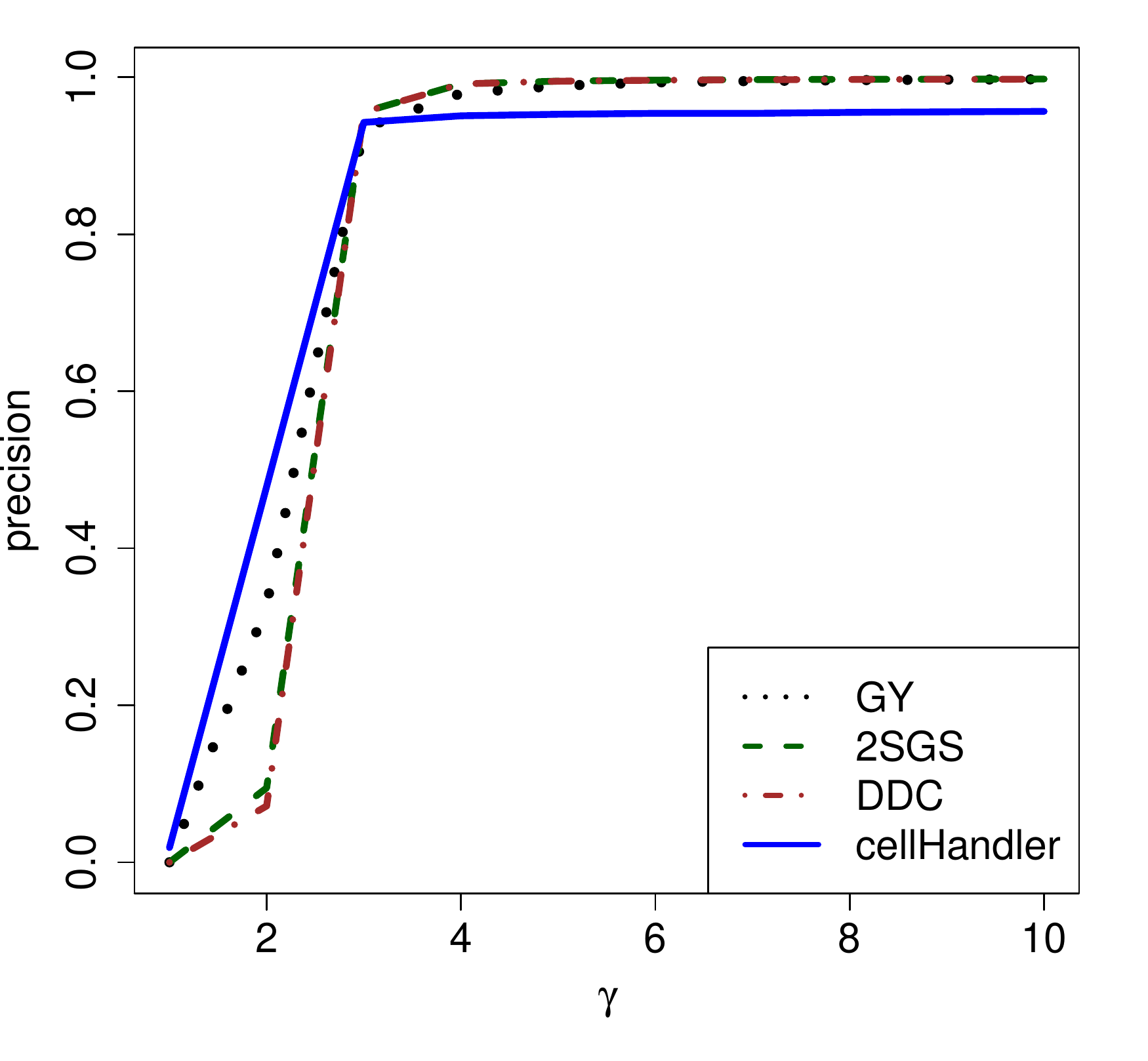}
\includegraphics[width = 0.43\textwidth]
   {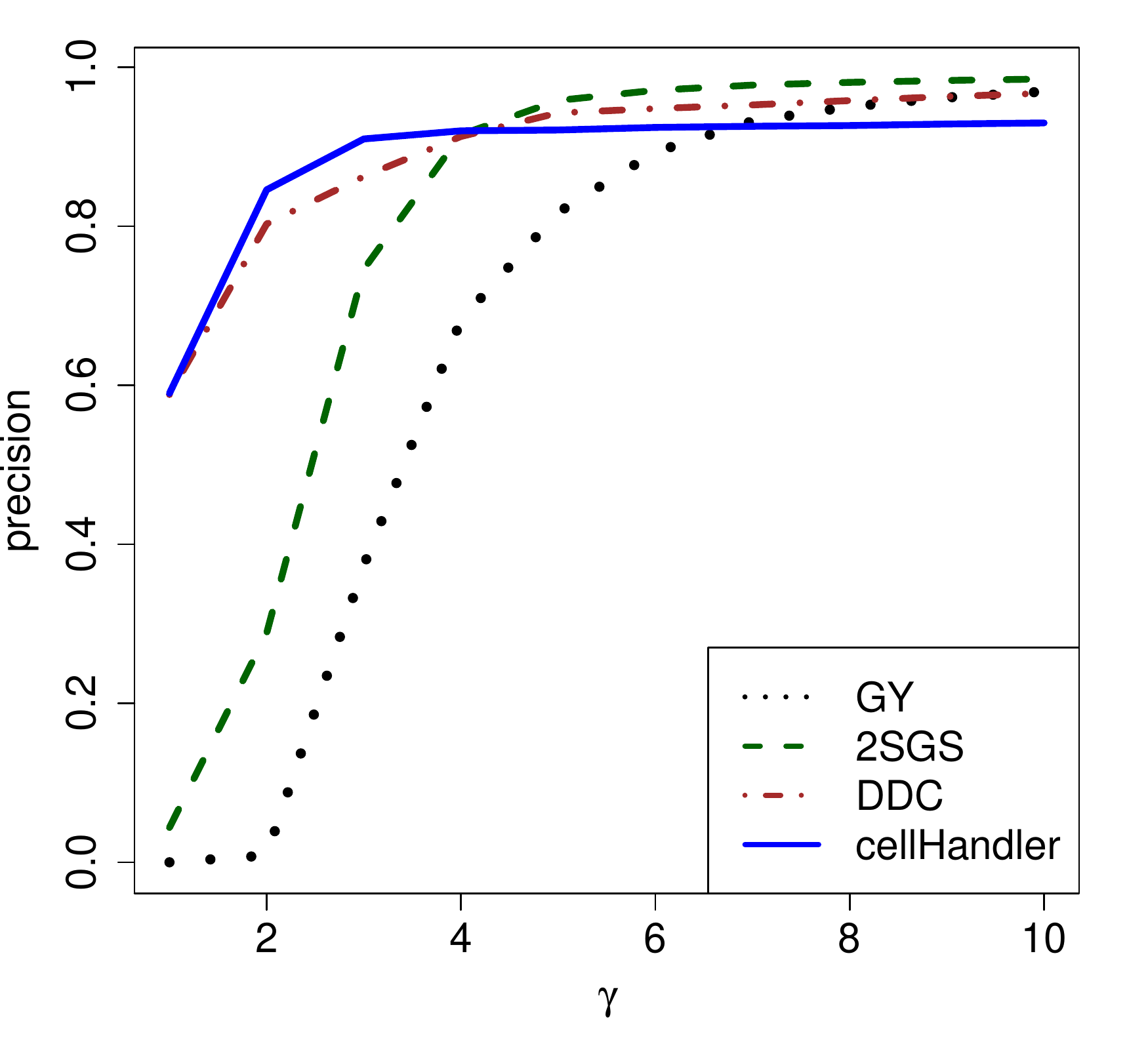}
\vskip-0.5cm
\includegraphics[width = 0.43\textwidth]
   {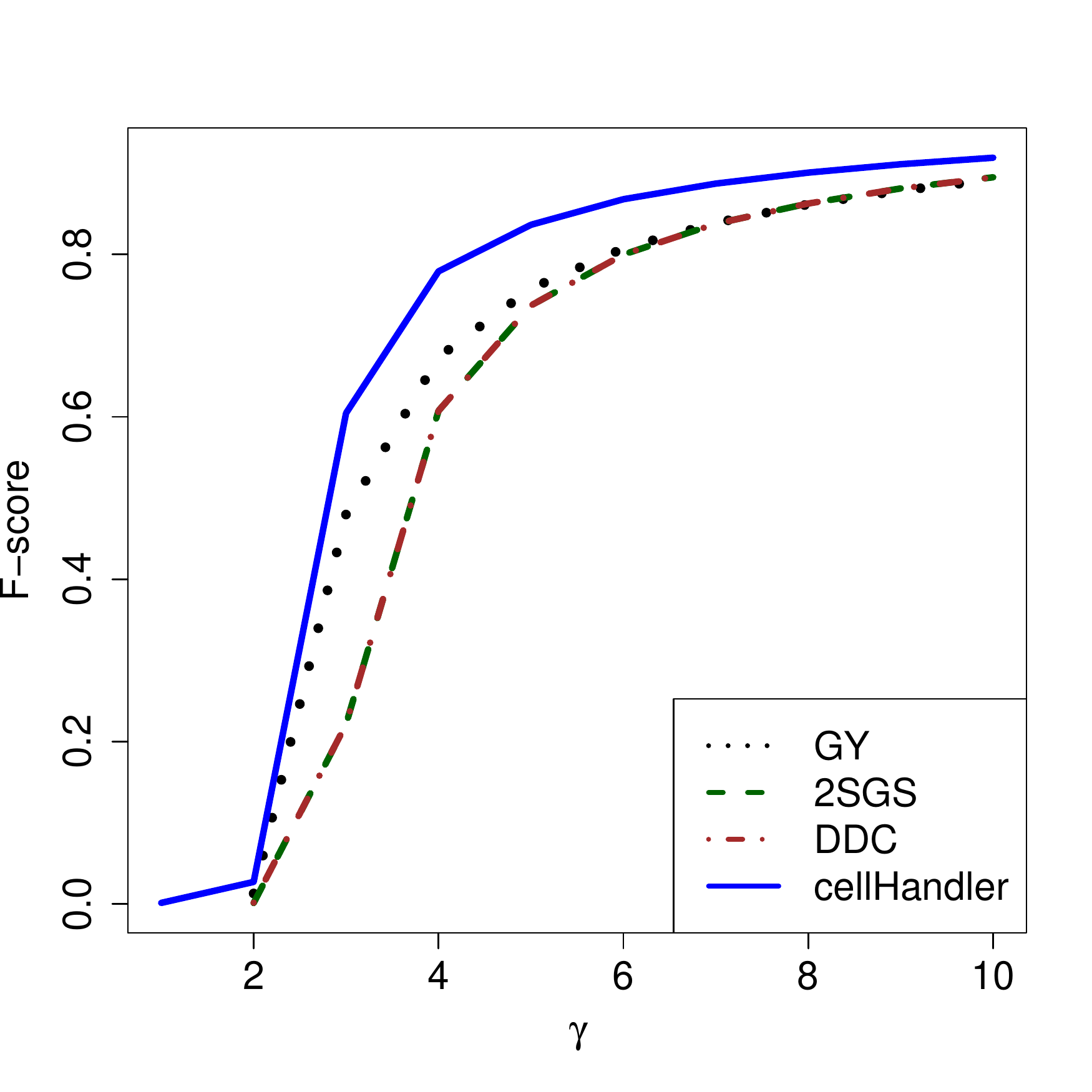}
\includegraphics[width = 0.43\textwidth]
   {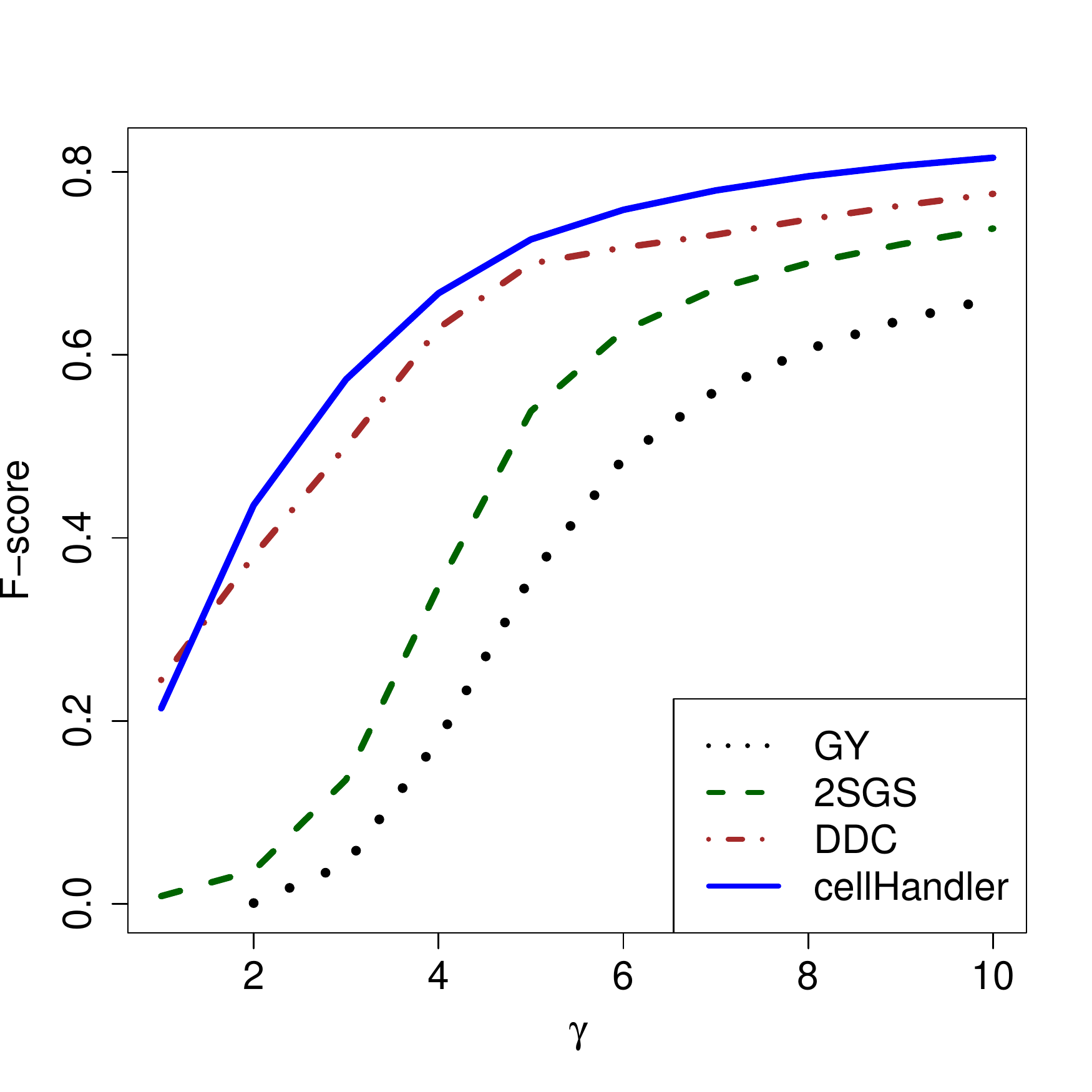}
	\vskip-0.2cm
\caption{Comparison of methods for detecting 
   cellwise outliers on data generated by 
	 the contaminated ALYZ model (left) and 
	 the contaminated  A09 model (right), 
	 for $n=400$ points in $d=20$ dimensions.
	 The plots show recall (top),
   precision (middle), and F-score (bottom).}
\label{fig:cellHandler_perout0.2}
\end{figure}

Figure 
\ref{fig:cellHandler_perout0.2} shows the 
performance of cellHandler on samples of 
size $n = 400$ in $d=20$ dimensions
with $\eps=20\%$ of cellwise outliers, 
using the covariance matrix estimated by 
the algorithm DDCW.DI 
described in Section \ref{sec:DIalgo}.
The other curves are from three existing 
techniques for flagging cells.
The first one is the univariate Gervini-Yohai 
filter (GY) specified in \citep{Agostinelli2015}.
The second is the multivariate 
DetectDeviatingCells (DDC) algorithm 
of \cite{Rousseeuw2018}, available in the 
\texttt{cellWise} package \citep{cellWise}
as the function \texttt{DDC}.
The third is the default filter of the 2SGS 
method in \citep{Leung2017}, which is a 
combination of a bivariate GY filter with DDC.
The top panels in Figure
\ref{fig:cellHandler_perout0.2} show the 
{\it recall}, which is the fraction of generated
cellwise outliers that are flagged as such.
The data in the left plot were generated by the
contaminated ALYZ model, and on the right by the
contaminated A09 model.
We see that cellHandler has the highest recall
at each $\gamma$. 
When $\gamma$ increases the cellwise outliers 
become marginally outlying, making them easier 
to flag.
The middle row of the figure shows the 
{\it precision}, which is the fraction of cells
flagged as outlying that were generated as such. 
We see that cellHandler does not have the 
best precision among competing methods at high 
$\gamma$, which is due to the tradeoff between
precision and recall.
Finally, the bottom row shows the F-score, also
called the Dice coefficient \citep{Dice1945},
which summarizes the performance of a binary 
classification through the harmonic mean of 
precision and recall. Based on this
summary measure, cellHandler performs best.

\section{Cellwise robust estimation of a covariance 
   matrix}
\label{sec:DI}
\subsection{Existing approaches}
\label{sec:existing}
The previous section described a method for flagging
cellwise outliers when the true center $\bmu$ and 
covariance matrix $\bSigma$ are known. 
Of course these are rarely known in practice, so 
they have to be estimated.
The center $\bmu$ can be estimated quite easily by 
applying a robust estimator (like the median) to 
each coordinate.
Estimating the covariance matrix $\bSigma$ is the
hard part. 
There exist several approaches to this problem. 

A popular technique is to compute robust covariances
between each pair of variables, and to assemble 
them in a matrix.
To estimate these pairwise covariances, 
\cite{Ollerer2015} and \cite{Croux2016} use
rank-based methods such as the Spearman and
normal scores correlations. 
\cite{tarr2016} instead propose to use the robust 
pairwise correlation estimator of 
\cite{Gnanadesikan1972} in combination with the 
robust scale estimator $Q_n$ of \cite{Rousseeuw1993}. 
As the resulting matrix is not necessarily positive
semidefinite (PSD), they then compute the nearest
PSD matrix by the algorithm of \cite{Higham2002}. 
All of these pairwise covariance estimators are fast 
to compute.
We will compare the performance of these methods
in Section \ref{sec:simulation}.

A second approach is the snipEM procedure 
proposed by \cite{Farcomeni2014} and implemented in 
the R package \texttt{snipEM} of 
\cite{Farcomeni2019}.
Its first step flags cellwise outliers in each 
variable separately using a boxplot rule, and then 
"snips" them, which means making them missing.  
The second step tries many interchanges that unsnip 
a randomly chosen snipped cell and at the same time 
snip a randomly chosen unsnipped cell, and only 
keeps an interchange when it increases the partial 
Gaussian likelihood.
This procedure is slower than the pairwise 
covariance approach.

The current state of the art to deal with complex 
cellwise outliers is the two-step generalized 
S-estimator (2SGS) of \cite{Agostinelli2015} and 
\cite{Leung2017} implemented in the R package 
\texttt{GSE} \citep{Leung2019}. 
In a first step, the method uses a filter
(called 2SGS in Figure 
\ref{fig:cellHandler_perout0.2} above) to detect  
cellwise outliers.
These cells are then set to missing, and the
generalized S-estimator of \cite{danilov2012} 
is run.
A short survey of cellwise robust covariance
estimators can be found in Sections 6.13 and
6.14 of \cite{Maronna:RobStat}.

\subsection{The detection-imputation algorithm}
\label{sec:DIalgo}
Our algorithm for constructing a cellwise robust
covariance matrix starts by standardizing the 
columns of the dataset as in the beginning of 
Section \ref{sec:ranking}.
Next, we compute initial estimators $\bhmu^0$ 
and $\bhSigma^0$.
For this we can use the 2SGS estimator of
\cite{Leung2017} described above.
We will also try a different initial estimator
called DDCW, which is a combination of the DDC
method \citep{Rousseeuw2018} and the wrapped 
covariance matrix of \cite{Raymaekers2018}.
This initial estimator is described in Section
\ref{A:DDCW} of the Appendix.
 
The detection-imputation (DI) algorithm then
alternates the D-step and the I-step, both
described below.\\

\vskip-0.4cm
\noindent \textbf{D-step: detecting outlying 
   cells across all rows.}\\
The D-step first applies the cellHandler method
of Section \ref{sec:cellHandler} to each row
$\bz_i'$ based on the estimates $\bhmu^{t-1}$ and 
$\bhSigma^{t-1}$ from the previous iteration step.
This way each row $\bz_i'$ gets a ranking of its 
cells $z_{ij}$\;. 
From the $\Delta_k$ in its path we construct a 
nonincreasing sequence of criterion values 
$C_{ij} \coloneqq 
 \max_{h \gs k(j)} \Delta_h$\;. 
If any cells $z_{ih}$ are missing (NA) 
these are put in front of the path
with $C_{ih} \coloneqq +\infty$.

Should some columns have too many flagged cells
(including NA's) it could become 
difficult to estimate a correlation between them, 
especially if the flagged sets overlap little.
Even worse, flagging all cells in a column
would remove all information about that variable.
Therefore, we impose a maximal number of flagged
cells in each column, including the NA's.
This number is $n\, \mbox{\texttt{maxCol}}$ where the 
input parameter \texttt{maxCol} is set to $25\%$ by 
default.
Note that this is a constraint on the columns, 
whereas we are flagging cells by row.
We resolve this with the following algorithm:

\noindent \hskip0.4cm - sort the criterion values 
  $C_{ij}$ of all cells in the matrix in decreasing 
	order;

\noindent \hskip0.4cm - walk down this list. 
  If a $C_{ij}$ lies 
  below the cutoff value $q$ we ``lock'' row $i$,
	i.e. no cells of row $i$ can be flagged any more.
	If $C_{ij} > q$ the cell is flagged, unless it
	belongs to a column which already has 
	$\,n\, \mbox{\texttt{maxCol}}\,$ flagged cells. 
	In the latter case, row $i$ is locked also.

\noindent This procedure yields a (possibly empty)
list of flagged cells in each row, 
which overall contains the most outlying cells 
subject to the \texttt{maxCol} constraint.\\

\vskip-0.4cm
\noindent \textbf{I-step: Re-estimate $\mu$ and 
  $\bSigma$}\,.\\
The I-step is basically one step of the EM
algorithm which considers the flagged cells as
missing.
However, it is computationally more efficient 
since it reuses results that are already 
available.
In each row, the set of flagged cells is one of
the active sets considered by LAR in 
cellHandler, so its coefficient $\bhtheta_1$ 
from Proposition \ref{prop:2} is known. 
This makes it trivial to impute the flagged cells,
so the E-step of EM requires no additional 
computation.
Next, $\bhmu^{t}$ and $\bhSigma^{t}$ are
computed as in the M-step, as described in more
detail in Section \ref{A:DI} of the Appendix.
This iterative procedure stops 
when both $\bhmu^{t} - \bhmu^{t-1}$ and
$\bhSigma^{t} - \bhSigma^{t-1}$ are small. 
At the end of the DI algorithm we unstandardize 
$\bhmu$ and $\bhSigma$ using the univariate 
location and scale estimates of the original 
data columns.

The time complexity of the DI algorithm is
$O(T n d^3)$ where $T$ is the number of iteration 
steps.
This is the same complexity as that of the classical 
EM algorithm for covariance estimation with missing 
data. 

Note that for the DI method to work, the initial
covariance matrix (whether 2SGS or DDCW) and those 
in all iteration steps need to be invertible. 
This requires that $n > d$, so for now the approach
does not allow for $d \geqslant n$. Possible
extensions are a topic for further research, and
would likely require penalization or other forms of
regularization.

\section{Measuring scatter matrix discrepancy}
\label{sec:discrepancy}
In the simulation in the next section we want to
measure how much an estimated scatter matrix
deviates from the true underlying positive definite 
(PD) scatter matrix.
For this we need a discrepancy measure for scatter
matrices.
Here we will construct a pre-existing discrepancy
measure from first principles, in order to dispel
a common misconception that this measure would
only make sense when the underlying data follow a
multivariate normal (Gaussian) distribution.

Suppose we want to measure how much a scatter matrix
$\bA$ deviates from a reference scatter matrix $\bB$,
where the $d \times d$ matrix $\bB$ is PD but 
$\bA$ only needs to be positive semidefinite (PSD). 
A simple measure of this type is
\begin{equation} \label{eq:Frobenius}
  ||\bA - \bB||_2 =  \Big( \sum_{i=1}^{d} 
  \sum_{j=1}^{d}\, (a_{ij}-b_{ij})^2\, 
	\Big)^{1/2}
\end{equation}
but it is insufficiently suited to our scatter 
matrix context, as it does not tell us whether 
$\bA$ is singular.
And this is important, since a singular scatter 
matrix $\bA$ cannot be used as an approximation
of $\bB$, for instance when computing a
Mahalanobis-style statistical distance as in 
\eqref{eq:MDreg} which requires the inverse matrix.

In order to stay in the land of scatter matrices
we instead compute 
$$\bC = \bB^{-1/2} \bA\, \bB^{-1/2}\;.$$
(The matrix $\bC$ can be seen as the scatter $\bA$ 
in the coordinate system where $\bB$
is sphered/whitened, since 
$\bB^{-1/2} \bB\, \bB^{-1/2} = \bI$.)
Note that $\bA = \bB$ if and only if $\bC = \bI$,
so we want to measure how far $\bC$ is from $\bI$
in a way that is relevant for scatter matrices.
Since the matrix $\bC$ is PSD its eigenvalues
are nonnegative, so we can denote them as
$\eta_1 \geqslant \ldots \geqslant \eta_d
\geqslant 0$.
We want the discrepancy measure to be zero if
all $\eta_j = 1$, to go to $+\infty$
when $\bC$ explodes in the sense that 
$\eta_1 \rightarrow +\infty$, and also when 
$\bC$ implodes, i.e. $\eta_d \rightarrow 0$.
Concentrating on a single eigenvalue $\eta$ we
want a continuous function $h(\eta)$ on all 
$\eta \geqslant 0$ with the properties 
$h(\eta) \geqslant 0$, $h(1) = 0$,
and $h(\eta) \rightarrow +\infty$ when 
$\eta \rightarrow +\infty$ or
$\eta \rightarrow 0$.
Many such functions can be constructed.
One of them is $h(\eta) = \eta - 1 -\log(\eta)$.
Note that $h(\eta) \geqslant 0$ since
$\log(\eta)$ is concave and $\eta - 1$ is its
tangent line at $\eta=1$.
The function $h$ decreases on $[0,1[$, reaches 
its minimum in 1 with $h(1)=0$, and 
increases on $]1,+\infty[$.
Therefore it makes sense to define the
discrepancy of $\bA$ relative to $\bB$ as
\begin{equation} \label{eq:discr}
   D(\bA,\bB) := \sum_{j=1}^d h(\eta_j)
	 = \sum_{j=1}^d (\eta_j - 1 -\log(\eta_j))
\end{equation}
which is nonnegative since each term is.
Note that $\bA = \bB$ is equivalent to
$D(\bA,\bB) = 0$, and that a singular $\bA$ 
attains $D(\bA,\bB) = +\infty$.

The entire construction of $D(\bA,\bB)$ above
only uses the PSD property of scatter matrices, 
and is not at all restricted to multivariate 
normally distributed data.
That confusion is due to the following
property:

\begin{proposition} \label{prop:4}
In the special case where $\bX$ and $\bY$ are
$d$-variate random vectors distributed as 
$\bX \sim N(\bzero,\bA)$ and 
$\bY \sim N(\bzero,\bB)$ in which both $\bA$
and $\bB$ are PD, the discrepancy
$D(\bA,\bB)$ coincides with the 
Kullback-Leibler divergence $\KL(\bX,\bY)$.
\end{proposition}

The proof is in Section \ref{A:KL} of 
the Appendix.
Note that when both $\bA$ and $\bB$ are PD
also $D(\bB,\bA)$ exists, but it does not
equal $D(\bA,\bB)$. Instead we obtain, 
along the same lines,
$D(\bB,\bA) = \sum_{j=1}^d h(1/\eta_j)$.
However, it is possible to symmetrize the
discrepancy $D(.,.)$ by replacing the $h$ in
\eqref{eq:discr} by a function $\tih$ for which
$\tih(1/\eta) = \tih(\eta)$ for all $\eta > 0$,
such as $\tih(\eta) = \eta + 1/\eta -2$.
One could also use the function 
$\tih(\eta) = |\log(\eta)|$ so 
$D(\bA,\bB)$ becomes the $L^1$ norm of
$(\log(\eta_1),\ldots,\log(\eta_d))$.

\section{Simulation results}
\label{sec:simulation}
We simulate the estimators of covariance 
matrices discussed in the previous section.
The data is generated as in Subsection 
\ref{subsec:cFsim}, with dimensions 
$d = 10$, $20$ and $40$. 
The fraction of contaminated cells is 
$\eps = 0.1, 0.2$ in which $\gamma$ varies from $1$ 
to $10$. 
In each replication we compute the discrepancy
\eqref{eq:discr} of the estimate 
$\widehat{\bSigma}$ from the underlying $\bSigma$,
and then average the discrepancy
 over all replications.
We show the results for $\eps = 0.2$, since this 
is the most challenging scenario. 
The results for $\eps = 0.1$ were qualitatively 
similar.

Figure \ref{fig:sim_d10} compares the proposed 
methods to  
the existing approaches described in Subsection 
\ref{sec:existing}, for $d=10$. 
Since $\eps = 0.2$ there are on average two
cellwise outliers per row.
Gaussian rank (Grank) and Spearman refer to the 
covariance matrices of \cite{Ollerer2015} and 
\cite{Croux2016} using those rank correlations.
The Gnanadesikan-Kettenring procedure of 
\cite{tarr2016} is labeled GKnpd. 
Next, the snipEM method of \cite{Farcomeni2014} and 
the 2SGS estimator of \cite{Leung2017} are plotted.
The method 2SGS.DI uses 2SGS as initial estimator 
followed by the new DI method of Section 
\ref{sec:DIalgo}.
Also the initial estimator DDCW described in
Section \ref{A:DDCW} of the Appendix is shown, 
as well as DI applied to it. 

\begin{figure}[!ht]
\centering
\vskip0.5cm
\begin{minipage}{0.49\linewidth}
  \centering 
    \textbf{ALYZ model, 20\% outliers, $\bm{d = 10}$}
	\includegraphics[width=0.9\textwidth]
	{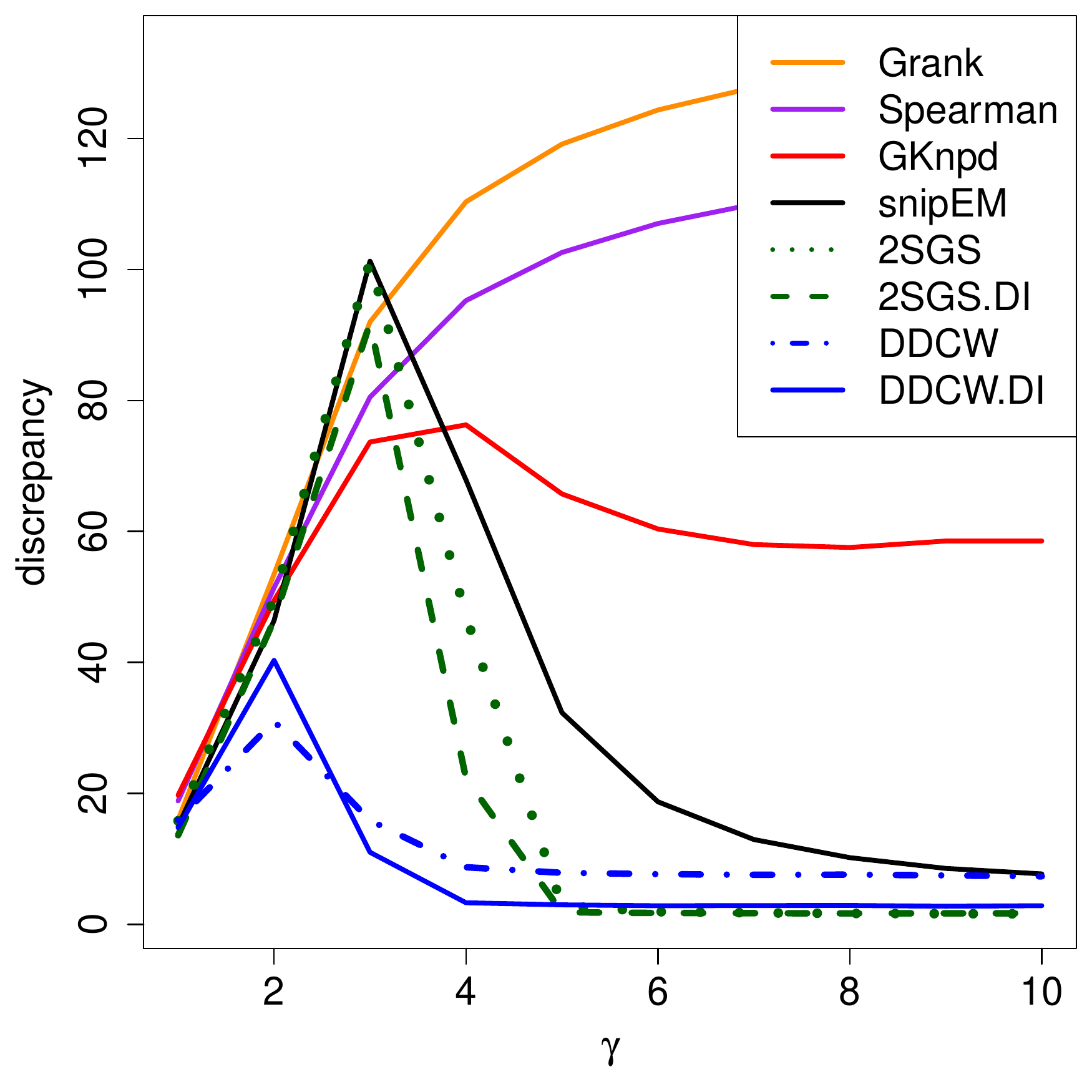} 
\end{minipage}
\begin{minipage}{0.49\linewidth}
  \centering
	  \textbf{A09 model, 20\% outliers, $\bm{d = 10}$}
  \includegraphics[width=0.9\textwidth]
	{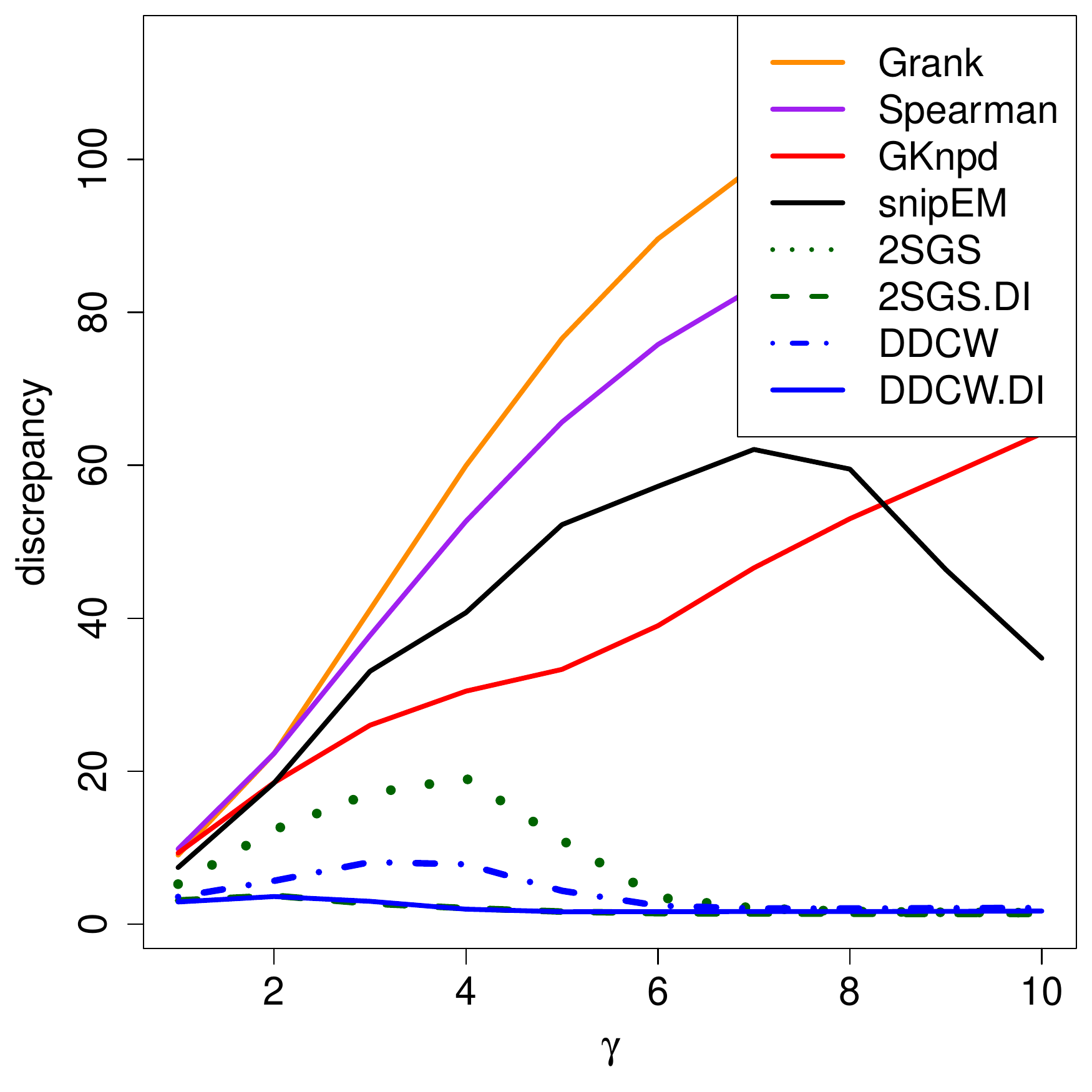} 
\end{minipage}
\vskip-0.3cm
\caption{Discrepancy $D(\widehat{\bSigma},\bSigma)$
         of estimated covariance matrices for 
         $d = 10$, $n = 100$.}
\label{fig:sim_d10}
\end{figure}

We see that the three pairwise methods Grank, 
Spearman and GKnpd pay for their fast computation
by a high discrepancy. 
The snipEM method does better for high $\gamma$,
in part because the boxplot rule in its first step 
snips marginally outlying cells.
The three pairwise methods do not use such a 
rule to flag marginally outlying cells, so high 
$\gamma$ values impact them more.
The state of the art method 2SGS does substantially
better, and is improved by applying DI to it, both
in the ALYZ and A09 models.
The same holds for DDCW and DDCW.DI.
Note that DI improves the results more under A09
than ALYZ, because A09 has bigger correlations
so DI has more opportunities to make a difference.

We now consider higher dimensions, starting with
$d=20$ in the top panels of Figure 
\ref{fig:KLdiv_highdim}. 
The curves of Grank, Spearman, GKnpd and snipEM
were much higher in this case, so we only show
the four best performing methods in order to see 
the differences between them. 
Also here the DI algorithm substantially improves
upon the initial estimators. The improvement is
largest under A09 which contains some high 
correlations.
For $d=40$ (bottom panels) we see similar patterns. 
	
\begin{figure}[!ht]
\centering
\vskip0.5cm
\begin{minipage}{0.49\linewidth}
  \centering 
    \textbf{ALYZ model, 20\% outliers, $\bm{d = 20}$}
	\includegraphics[width=0.9\textwidth]
	  {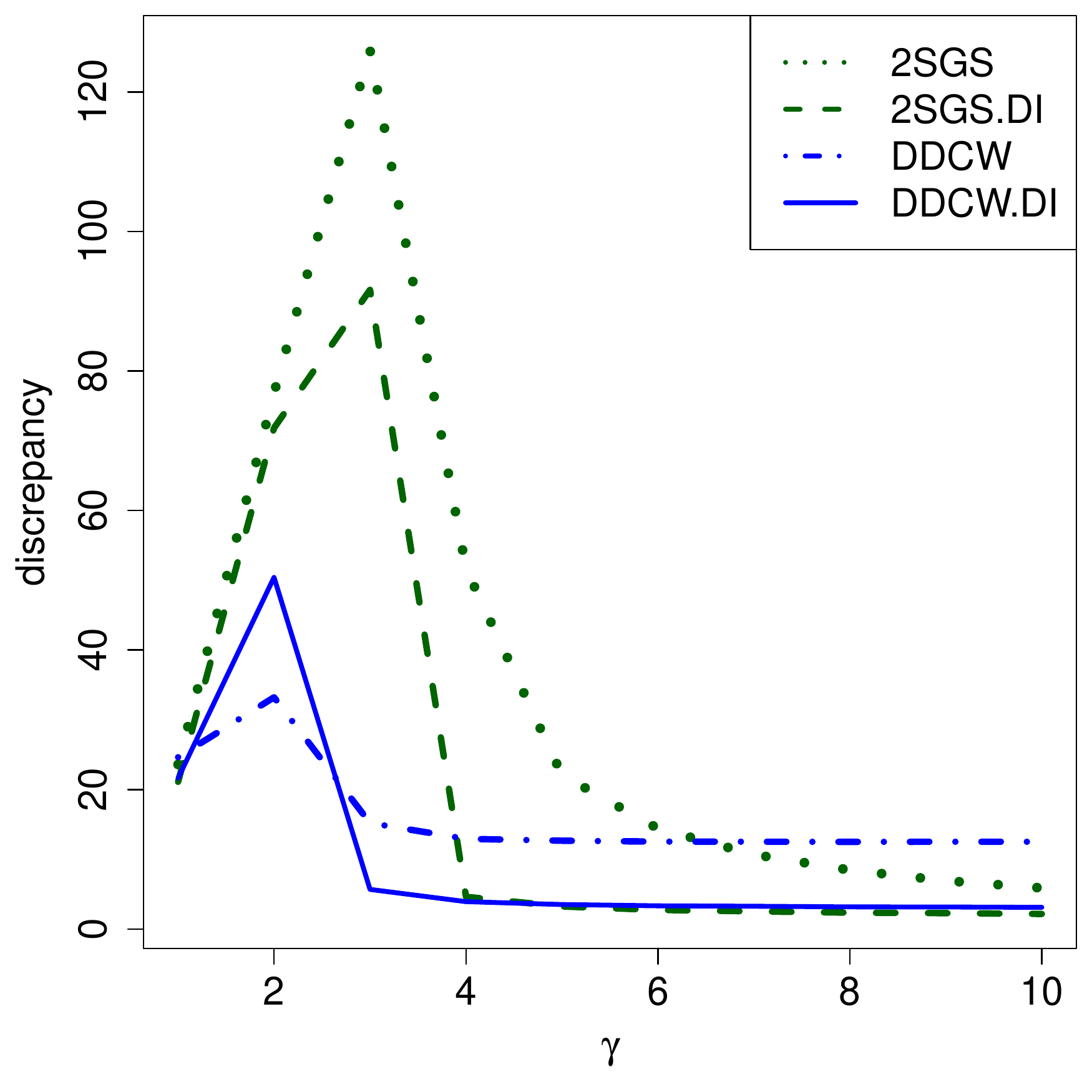} 
\end{minipage}
\begin{minipage}{0.49\linewidth}
  \centering
	  \textbf{A09 model, 20\% outliers, $\bm{d = 20}$}
  \includegraphics[width=0.9\textwidth]
	  {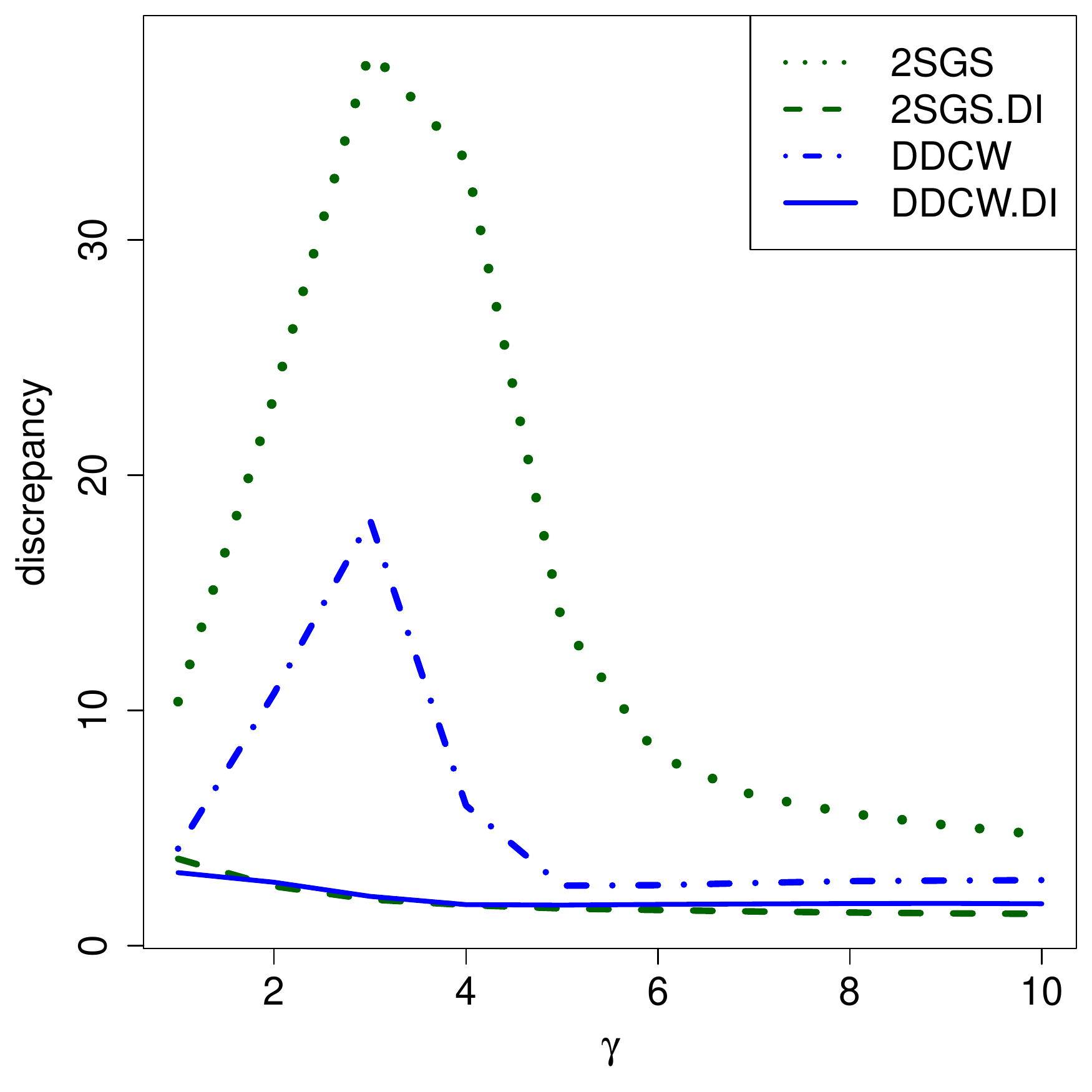} 
\end{minipage}
\vskip0.3cm
\begin{minipage}{0.49\linewidth}
\centering 
    \textbf{ALYZ model, 20\% outliers, $\bm{d = 40}$}
	\includegraphics[width=0.9\textwidth]
	  {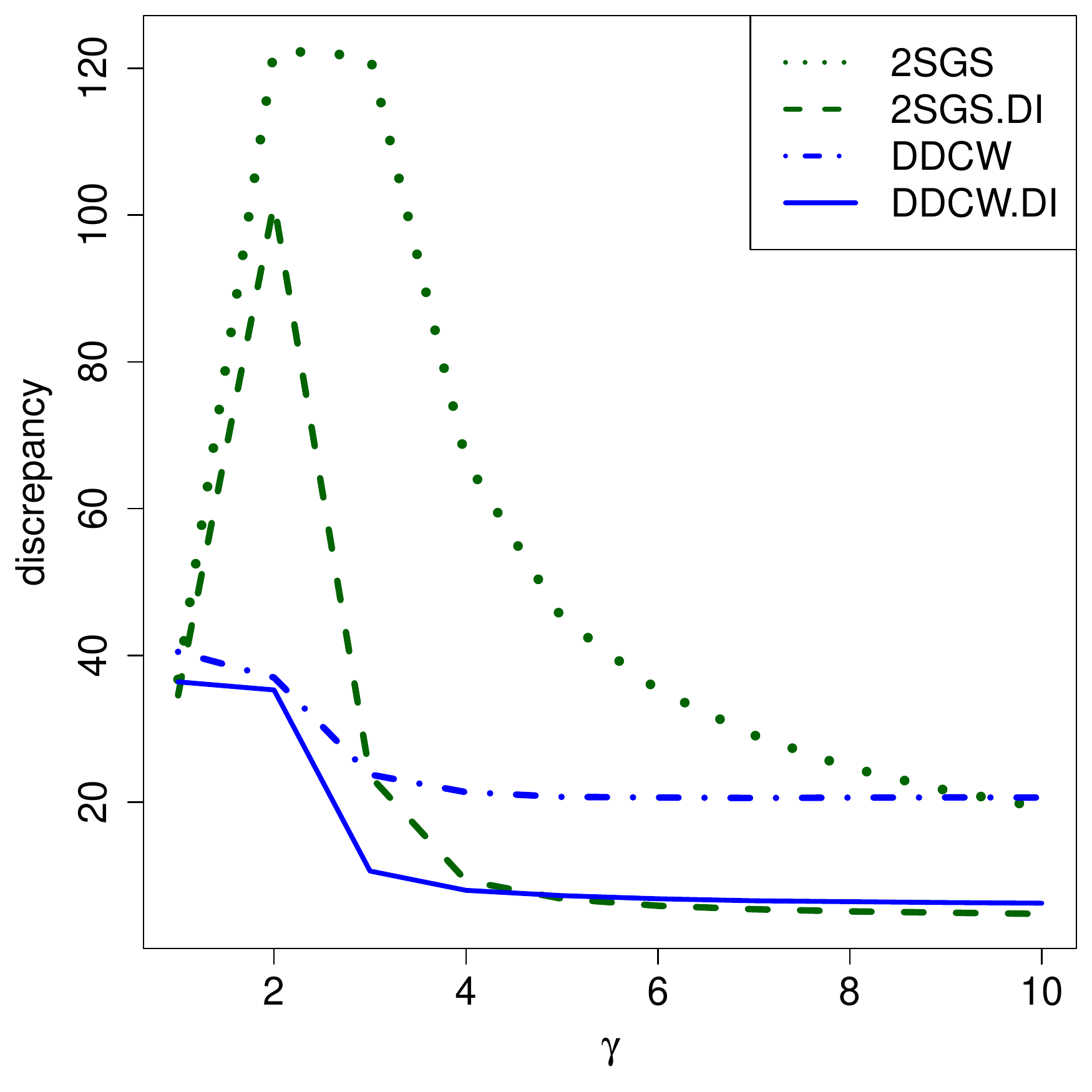} 
\end{minipage}
\begin{minipage}{0.49\linewidth}
  \centering
	  \textbf{A09 model, 20\% outliers, $\bm{d = 40}$}
  \includegraphics[width=0.9\textwidth]
	  {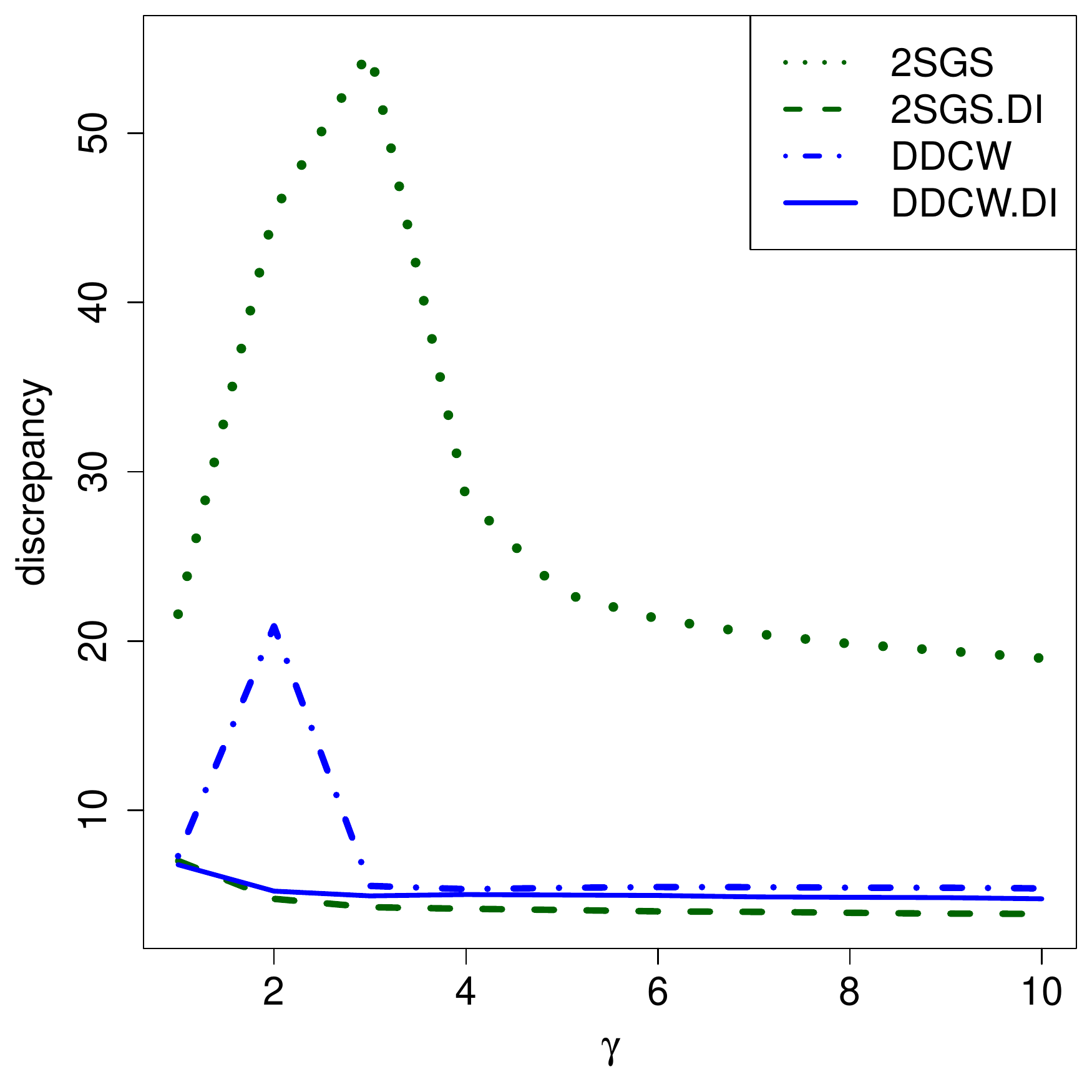} 
\end{minipage}
\vskip-0.3cm
\caption{Discrepancy $D(\widehat{\bSigma},\bSigma)$
  given by \eqref{eq:discr} of estimated covariance
  matrices for $d = 20$ and $n = 400$ (top panels)
	and for $d = 40$ and $n = 800$ (bottom panels).}
\label{fig:KLdiv_highdim}
\end{figure}

Section \ref{A:sim} of the Appendix
shows the results of a simulation in which the
data are contaminated by $10\%$ of cellwise 
outliers generated as above, plus $10\%$ of 
rowwise outliers.
In this particular setting "rowwise outliers" 
refers to rows in which all cells are contaminated 
in the same way as before, that is, rows with $d$ 
cellwise outliers.
The initial estimators 2SGS and DDCW attempt to
downweight or discard such rows.
The results are qualitatively similar to those in
Figures \ref{fig:sim_d10} and 
\ref{fig:KLdiv_highdim}.

\section{Example: volatile organic compounds in children}
\label{sec:example}

We study a dataset of volatile organic 
compounds (VOCs) in human urinary samples. 
The data was taken from the publicly available 
website of the National Health and Nutrition 
Examination Survey \citep{NHANES}, 
using the most recent available epoch.
Such VOC metabolites are commonly monitored since 
chronic exposure to high levels of some VOCs can 
lead to a number of health problems such as 
cancer and neurocognitive dysfunction. 
The original dataset consists of 29 VOC metabolites, 
but we focus on a subset of 16 variables obtained 
by removing columns with a lot of missing 
values and/or zero median absolute deviation.
Section \ref{A:VOCs} in the Appendix
contains a table with the VOCs analyzed. 
In order to obtain a relatively homogeneous subset, 
we selected the data for children aged 10 or younger.
The final dataset contained 512 subjects.
We log-transformed the concentrations to make the 
variables roughly Gaussian 
(apart from possible outliers).

We estimated the covariance matrix of
the data by the DI algorithm, starting from 
the DDCW initial estimator.
The algorithm converged after 7 steps. 
Using the resulting covariance estimate we ran 
the cellHandler algorithm with cutoff
$\sqrt{\chi^2_{1,0.99}} \approx 2.57$
to detect outlying cells. 
The corresponding cellmap of the first 20 
children in the list was shown as
Figure \ref{fig:cellmap} in the introduction.
Each row of the cellmap corresponds to a child, 
with inlying cells colored yellow.
Red colored cells indicate that their value is 
higher than predicted given the inlying cells 
of that row, while blue cells indicate lower 
than predicted values. 
The more extreme the residual, the more 
intense the color.

One variable that stood out was URXCYM 
(N-Acetyl-S-(2-cyanoethyl)-L-cysteine) in 
which cellHandler indicated $11\%$ of large 
cell residuals.
This was particularly striking since that 
variable had fewer than $2\%$ of marginal
outliers using the same cutoff
$\sqrt{\chi^2_{1,0.99}}$
on the absolute standardized values,
and these were rather nearby
(note that even for perfectly Gaussian data
there would already be $1\%$ of absolute 
standardized values above this cutoff).
Figure \ref{fig:ZresVersusZ} plots the
cell residuals (which are zero for cells
that were not flagged) versus the robustly
standardized marginal values, with the 
cutoffs indicated by horizontal and 
vertical lines.
Most of the outlying cellwise residuals
correspond to inlying marginal values. 
These children have extreme URXCYM values
relative to their other VOCs.

\begin{figure}[!ht]
\center
\vskip0.4cm
\includegraphics[width = 0.55\textwidth]
  {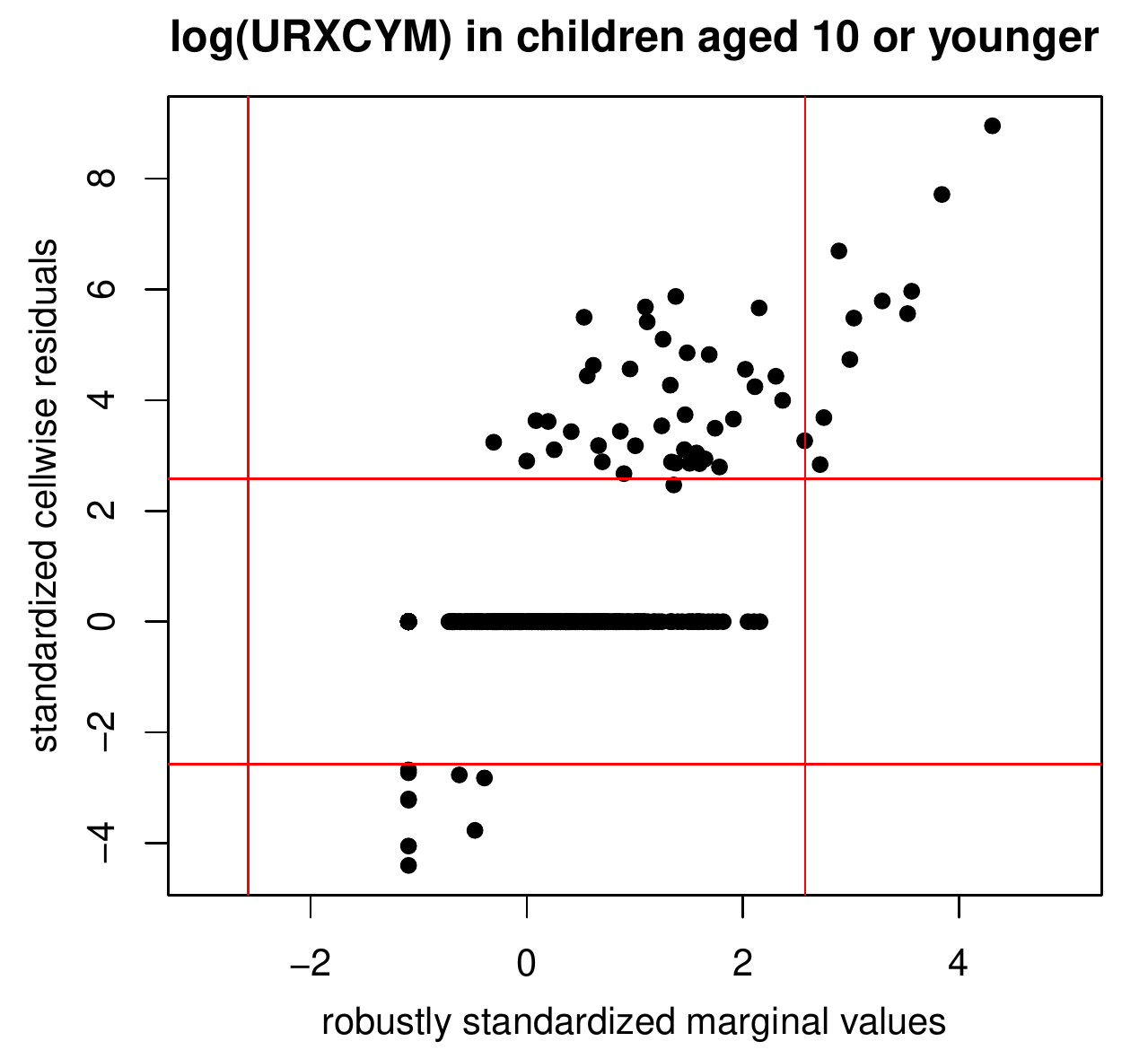}
\vskip-0.2cm
\caption{Plot of standardized cell residuals
of log(URXCYM) obtained by cellHandler, versus
the robustly standardized values of
log(URXCYM) on its own.}
\label{fig:ZresVersusZ}
\end{figure}

Interestingly, URXCYM is a well known 
biomarker for identifying smokers among adults, 
see e.g. \cite{Chen2019}, since it typically 
results from the metabolization of acrylonitrile, 
a volatile liquid present in tobacco smoke. 
But in this example we are studying children, 
who are not supposed to smoke.
In search of an explanation we combined the VOC data 
with the questionnaire data available on the same 
website \citep{NHANES}.
Among many other things, these data contain 
information on the smoking status of the adults
(usually parents) in the same household.
These fell into four categories: only nonsmoking
adults, smoking adults who do not smoke inside
the home, one adult smoking in the home, and 
two adults smoking in the home.
The blue curve in Figure \ref{fig:VOCS} shows
the percentage of children with URXCYM cell
residuals above the cutoff, in each of these
categories.
They go from 4.7\% in households with only
nonsmoking adults up to 72.7\% in homes where
two adults smoke, indicating that passive
smoking has a measurable effect on children.
On the other hand, if we were to look only at
the marginal URXCYM values (red curve) no such
effect is visible.

\begin{figure}[!ht]
\center
\vskip0.2cm
\includegraphics[width = 0.58\textwidth]
  {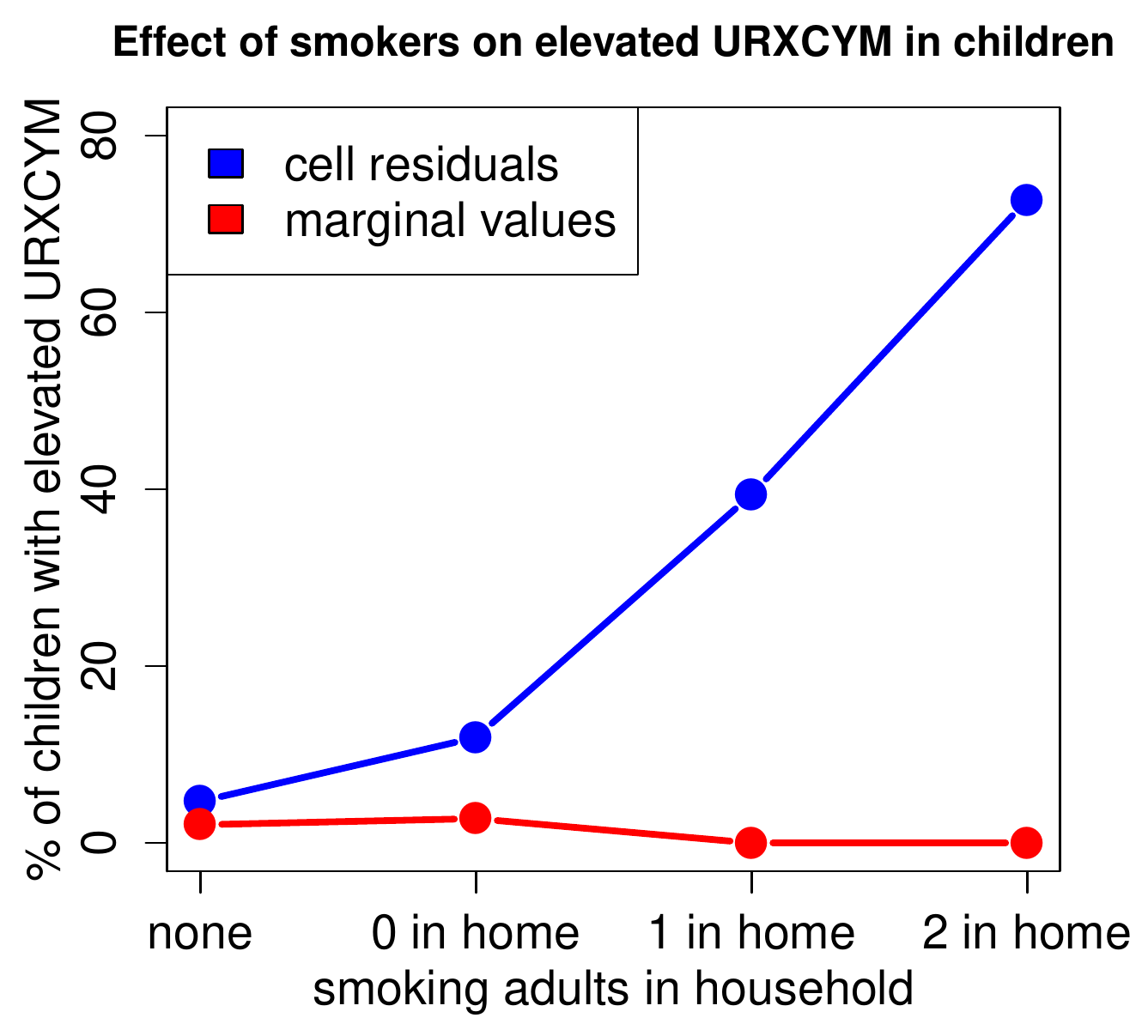}
\vskip-0.2cm
\caption{The blue curve shows the percentage 
of elevated URXCYM cell residuals in function 
of the smoking status of adult family members. 
The red curve shows the percentage of elevated
marginal URCYM values.}
\label{fig:VOCS}
\end{figure}

The example shows that the effect of 
exposing children to tobacco smoke could be  
underestimated when only performing 
univariate analyses on biomarkers. 
This illustrates that cell residuals obtained 
by cellHandler may add valuable information 
to a dataset.

\section{Conclusion}
The proposed cellHandler method is the first
to detect cellwise outliers 
based on robust estimates of location and
covariance.
It is also a major component of
the detection-imputation (DI) algorithm that
computes such cellwise robust estimates.
Note that both methods can deal with 
missing values in the data, since these are
imputed along the way.

The performance of cellHandler and DI was 
illustrated by simulation.
A real example illustrated that the common
medical practice of comparing individual
biomarkers to their tolerance limits
can benefit from the use of cellwise 
residuals. \\

\noindent \textbf{Acknowledgment.} 
This research	was funded by 
projects of Internal Funds KU Leuven.\\	

\noindent \textbf{Software Availability.}
The methods proposed here are
available as the functions\linebreak 
\texttt{cellHandler} and \texttt{DI} in the
\texttt{R} package \texttt{cellWise} 
\citep{cellWise} on CRAN.\\ 
It also contains a vignette which reproduces 
the example in Section \ref{sec:example}.\\


\begin{thebibliography}{}

\bibitem[Agostinelli et~al., 2015]{Agostinelli2015}
Agostinelli, C., Leung, A., Yohai, V.~J., and Zamar, R.~H. (2015).
\newblock Robust estimation of multivariate location and scatter in the
  presence of cellwise and casewise contamination.
\newblock {\em TEST}, 24:441--461, ISSN: {\ttfamily 1863-8260}, DOI:
  \href{https://dx.doi.org/10.1007/s11749-015-0450-6}{\ttfamily
  10.1007/s11749-015-0450-6}.

\bibitem[Alqallaf et~al., 2009]{alqallaf2009}
Alqallaf, F., Van~Aelst, S., Yohai, V.~J., and Zamar, R.~H. (2009).
\newblock Propagation of outliers in multivariate data.
\newblock {\em The Annals of Statistics}, 37:311--331, DOI:
  \href{https://dx.doi.org/10.1214/07-AOS588}{\ttfamily 10.1214/07-AOS588}.

\bibitem[Chandola et~al., 2009]{Chandola2009}
Chandola, V., Banerjee, A., and Kumar, V. (2009).
\newblock Anomaly detection: A survey.
\newblock {\em ACM Computing Surveys}, 41:1--58, DOI:
  \href{https://dx.doi.org/10.1145/1541880.1541882}{\ttfamily
  10.1145/1541880.1541882}.

\bibitem[Chen et~al., 2019]{Chen2019}
Chen, M., Carmella, S.~G., Sipe, C., Jensen, J., Luo, X., Le, C.~T., Murphy,
  S.~E., Benowitz, N.~L., McClernon, F.~J., Vandrey, R., Allen, S.~S.,
  Denlinger-Apte, R., Cinciripini, P.~M., Strasser, A.~A., al’Absi, M.,
  Robinson, J.~D., Donny, E.~C., Hatsukami, D., and Hecht, S.~S. (2019).
\newblock Longitudinal stability in cigarette smokers of urinary biomarkers of
  exposure to the toxicants acrylonitrile and acrolein.
\newblock {\em PLOS ONE}, 14(1):1--13, DOI:
  \href{https://dx.doi.org/10.1371/journal.pone.0210104}{\ttfamily
  10.1371/journal.pone.0210104}.

\bibitem[Croux and {\"O}llerer, 2016]{Croux2016}
Croux, C. and {\"O}llerer, V. (2016).
\newblock Robust and sparse estimation of the inverse covariance matrix using
  rank correlation measures.
\newblock In Agostinelli, C., Basu, A., Filzmoser, P., and Mukherjee, D.,
  editors, {\em Recent Advances in Robust Statistics: Theory and Applications},
  pages 35--55, New Delhi. Springer India, DOI:
  \href{https://dx.doi.org/10.1007/978-81-322-3643-6_3}{\ttfamily
  10.1007/978-81-322-3643-6\_3}.

\bibitem[Danilov, 2010]{danilov2010}
Danilov, M. (2010).
\newblock {\em Robust estimation of multivariate scatter in non-affine
  equivariant scenarios}.
\newblock PhD thesis, University of British Columbia.

\bibitem[Danilov et~al., 2012]{danilov2012}
Danilov, M., Yohai, V.~J., and Zamar, R.~H. (2012).
\newblock Robust estimation of multivariate location and scatter in the
  presence of missing data.
\newblock {\em Journal of the American Statistical Association},
  107:1178--1186, DOI:
  \href{https://dx.doi.org/10.1080/01621459.2012.699792}{\ttfamily
  10.1080/01621459.2012.699792}.

\bibitem[Debruyne et~al., 2019]{Debruyne2019}
Debruyne, M., H\"oppner, S., Serneels, S., and Verdonck, T. (2019).
\newblock Outlyingness: {W}hich variables contribute most?
\newblock {\em Statistics and Computing}, 29:707--723, DOI:
  \href{https://dx.doi.org/10.1007/s11222-018-9831-5}{\ttfamily
  10.1007/s11222-018-9831-5}.

\bibitem[Dice, 1945]{Dice1945}
Dice, L.~R. (1945).
\newblock Measures of the amount of ecologic association between species.
\newblock {\em Ecology}, 26:297--302, DOI:
  \href{https://dx.doi.org/10.2307/1932409}{\ttfamily 10.2307/1932409}.

\bibitem[Efron et~al., 2004]{efron2004}
Efron, B., Hastie, T., Johnstone, I., and Tibshirani, R. (2004).
\newblock Least angle regression.
\newblock {\em The Annals of Statistics}, 32:407--499, DOI:
  \href{https://dx.doi.org/10.1214/009053604000000067}{\ttfamily
  10.1214/009053604000000067}.

\bibitem[Farcomeni, 2014]{Farcomeni2014}
Farcomeni, A. (2014).
\newblock Robust constrained clustering in presence of entry-wise outliers.
\newblock {\em Technometrics}, 56(1):102--111, DOI:
  \href{https://dx.doi.org/10.1080/00401706.2013.826148}{\ttfamily
  10.1080/00401706.2013.826148}.

\bibitem[Farcomeni and Leung, 2019]{Farcomeni2019}
Farcomeni, A. and Leung, A. (2019).
\newblock {\em Package snipEM: Snipping Methods for Robust Estimation and
  Clustering}.
\newblock CRAN, R package version 1.0.1,
  \url{https://CRAN.R-project.org/package=snipEM}.

\bibitem[Gnanadesikan and Kettenring, 1972]{Gnanadesikan1972}
Gnanadesikan, R. and Kettenring, J.~R. (1972).
\newblock Robust estimates, residuals, and outlier detection with multiresponse
  data.
\newblock {\em Biometrics}, 28:81--124, ISSN: {\ttfamily 0006341X,} {\ttfamily
  15410420}, \url{http://www.jstor.org/stable/2528963}.

\bibitem[Hastie and Efron, 2015]{HastieLars}
Hastie, T. and Efron, B. (2015).
\newblock {\em Package lars: Least Angle Regression, Lasso and Forward
  Stagewise}.
\newblock CRAN, R package version 1.2,
  \url{https://CRAN.R-project.org/package=lars}.

\bibitem[Higham, 2002]{Higham2002}
Higham, N.~J. (2002).
\newblock {Computing the nearest correlation matrix - a problem from finance}.
\newblock {\em IMA Journal of Numerical Analysis}, 22:329--343, ISSN:
  {\ttfamily 0272-4979}, DOI:
  \href{https://dx.doi.org/10.1093/imanum/22.3.329}{\ttfamily
  10.1093/imanum/22.3.329}.

\bibitem[Leung et~al., 2019]{Leung2019}
Leung, A., Danilov, M., Yohai, V., and Zamar, R. (2019).
\newblock {\em Package GSE: Robust Estimation in the Presence of Cellwise and
  Casewise Contamination and Missing Data}.
\newblock CRAN, R package version 4.2,
  \url{https://CRAN.R-project.org/package=GSE}.

\bibitem[Leung et~al., 2017]{Leung2017}
Leung, A., Yohai, V., and Zamar, R. (2017).
\newblock Multivariate location and scatter matrix estimation under cellwise
  and casewise contamination.
\newblock {\em Computational Statistics \& Data Analysis}, 111:59--76, ISSN:
  {\ttfamily 0167-9473}, DOI:
  \href{https://dx.doi.org/10.1016/j.csda.2017.02.007}{\ttfamily
  10.1016/j.csda.2017.02.007}.

\bibitem[Little and Rubin, 1987]{Little:EM}
Little, R. and Rubin, D. (1987).
\newblock {\em Statistical Analysis with Missing Data}.
\newblock Wiley-Interscience, New York.

\bibitem[Maronna et~al., 2019]{Maronna:RobStat}
Maronna, R., Martin, D., Yohai, V., and Salibi\'an-Barrera, M. (2019).
\newblock {\em Robust Statistics: Theory and Methods, Second Edition}.
\newblock Wiley, New York.

\bibitem[NHANES, 2019]{NHANES}
NHANES (2019).
\newblock {N}ational {H}ealth and {N}utrition {E}xamination {S}urvey {D}ata.
\newblock Hyattsville, MD: U.S. Department of Health and Human Services,
  Centers for Disease Control and Prevention.
  \url{https://wwwn.cdc.gov/nchs/nhanes/}.

\bibitem[{\"O}llerer and Croux, 2015]{Ollerer2015}
{\"O}llerer, V. and Croux, C. (2015).
\newblock Robust high-dimensional precision matrix estimation.
\newblock In Nordhausen, K. and Taskinen, S., editors, {\em Modern
  Nonparametric, Robust and Multivariate Methods: Festschrift in Honour of
  Hannu Oja}, pages 325--350. Springer International Publishing, Cham, DOI:
  \href{https://dx.doi.org/10.1007/978-3-319-22404-6_19}{\ttfamily
  10.1007/978-3-319-22404-6\_19}.

\bibitem[Petersen and Pedersen, 2012]{Petersen2012}
Petersen, K.~B. and Pedersen, M.~S. (2012).
\newblock {\em The Matrix Cookbook}.
\newblock Technical University of Denmark,
  \url{http://www2.imm.dtu.dk/pubdb/pubs/3274-full.html}.

\bibitem[Raymaekers and Rousseeuw, 2019]{Raymaekers2018}
Raymaekers, J. and Rousseeuw, P.~J. (2019).
\newblock Fast robust correlation for high-dimensional data.
\newblock DOI:
  \href{https://dx.doi.org/10.1080/00401706.2019.1677270}{\ttfamily
  10.1080/00401706.2019.1677270}.
\newblock {\it Technometrics}, published online, open access.

\bibitem[Raymaekers et~al., 2020]{cellWise}
Raymaekers, J., Rousseeuw, P.~J., {Van den Bossche}, W., and Hubert, M. (2020).
\newblock {\em Package cellWise: Analyzing Data with Cellwise Outliers}.
\newblock CRAN, R package version 2.2.2,
  \url{https://CRAN.R-project.org/package=cellWise}.

\bibitem[Rousseeuw and Croux, 1993]{Rousseeuw1993}
Rousseeuw, P.~J. and Croux, C. (1993).
\newblock Alternatives to the median absolute deviation.
\newblock {\em Journal of the American Statistical Association}, 88:1273--1283,
  DOI: \href{https://dx.doi.org/10.1080/01621459.1993.10476408}{\ttfamily
  10.1080/01621459.1993.10476408}.

\bibitem[Rousseeuw and {Van den Bossche}, 2018]{Rousseeuw2018}
Rousseeuw, P.~J. and {Van den Bossche}, W. (2018).
\newblock Detecting deviating data cells.
\newblock {\em Technometrics}, 60:135--145, DOI:
  \href{https://dx.doi.org/10.1080/00401706.2017.1340909}{\ttfamily
  10.1080/00401706.2017.1340909}.
\newblock Open access.

\bibitem[Tarr et~al., 2016]{tarr2016}
Tarr, G., Müller, S., and Weber, N. (2016).
\newblock Robust estimation of precision matrices under cellwise contamination.
\newblock {\em Computational Statistics \& Data Analysis}, 93:404--420, ISSN:
  {\ttfamily 0167-9473}, DOI:
  \href{https://dx.doi.org/10.1016/j.csda.2015.02.005}{\ttfamily
  10.1016/j.csda.2015.02.005}.

\bibitem[Templ et~al., 2020]{templ2011}
Templ, M., Hron, K., and Filzmoser, P. (2020).
\newblock {\em Package robCompositions: Compositional Data Analysis}.
\newblock CRAN, R package version 2.2.1,
  \url{https://CRAN.R-project.org/package=robCompositions}.

\bibitem[Tibshirani, 1996]{tibs:lasso}
Tibshirani, R. (1996).
\newblock Regression shrinkage and selection via the lasso.
\newblock {\em Journal of the Royal Statistical Society Series B}, 58:267--288,
  \url{https://www.jstor.org/stable/2346178}.

\bibitem[Unwin, 2019]{Unwin:O3}
Unwin, A. (2019).
\newblock Multivariate outliers and the {O3} plot.
\newblock {\em Journal of Computational and Graphical Statistics}, 28:635--643,
  DOI: \href{https://dx.doi.org/10.1080/10618600.2019.1575226}{\ttfamily
  10.1080/10618600.2019.1575226}.

\bibitem[{Van Aelst} et~al., 2011]{VanAelst2011}
{Van Aelst}, S., Vandervieren, E., and Willems, G. (2011).
\newblock Stahel-{D}onoho estimators with cellwise weights.
\newblock {\em Journal of Statistical Computation and Simulation}, 81(1):1--27,
  DOI: \href{https://dx.doi.org/10.1080/00949650903103873}{\ttfamily
  10.1080/00949650903103873}.

\end{thebibliography}



\newpage

\begin{appendix}

\section{Proof of Proposition 2}
\label{A:Estep}
\begin{proof}
We first split up the relevant matrices in blocks.
Denote 
$$\bSigma^{-1} = \begin{bmatrix}
\bSigma_{11}^* & \bSigma_{12}^* \\
\bSigma_{21}^* & \bSigma_{22}^*
\end{bmatrix}
\quad \quad \mbox{  and  } \quad \quad
\bSigma^{-1/2} = \begin{bmatrix}
\btSigma_{11} & \btSigma_{12} \\
\btSigma_{21} & \btSigma_{22}
\end{bmatrix}\;\;.$$
Let the $k$-variate $\bhtheta$ be the solution to 
the ordinary least squares regression problem
\begin{equation*}
\argmin_{\btheta} ||\bSigma^{-1/2}(\bz - \bmu)- 
   (\bSigma^{-1/2})_{\bdot 1} \btheta||_2^2 
\end{equation*}
where $(\bSigma^{-1/2})_{\bdot 1}$ denotes the 
first $k$ columns of the matrix $\bSigma^{-1/2}$.\\	
Then we know that 
\begin{align*}
\bhtheta &= \left(\begin{bmatrix}
\btSigma_{11}  \\
\btSigma_{21}
\end{bmatrix}'\, \begin{bmatrix}
\btSigma_{11}  \\
\btSigma_{21}
\end{bmatrix} \right)^{-1} \begin{bmatrix}
\btSigma_{11}  \\
\btSigma_{21}
\end{bmatrix}'\, \bSigma^{-1/2}(\bz- \bmu)\\
& = (\bSigma_{11}^*)^{-1} \begin{bmatrix}
\btSigma_{11}  \\
\btSigma_{21}
\end{bmatrix}'\, \bSigma^{-1/2}(\bz-\bmu).
\end{align*} 
Now observe that
\begin{align*}
\bz_1 - \bhtheta &= \bz_1 - (\bSigma_{11}^*)^{-1} 
\begin{bmatrix}
\btSigma_{11}  \\
\btSigma_{21}
\end{bmatrix}'\, \bSigma^{-1/2}(\bz-\bmu) \\
&= \bz_1 - (\bSigma_{11}^*)^{-1} [\bSigma_{11}^* \; 
   \bSigma_{12}^*][\bz_1-\bmu_1 \;\; \bz_2-\bmu_2]\\
&= \bz_1 - [I_k \; \; (\bSigma_{11}^*)^{-1}
   \bSigma_{12}^*][\bz_1-\bmu_1 \;\; \bz_2-\bmu_2]\\
&= \bz_1 - (\bz_1 - \bmu_1) - 
   (\bSigma_{11}^*)^{-1}\bSigma_{12}^*(z_2-\bmu_2)\\
&= \bmu_1 + \bSigma_{12}\bSigma_{22}^{-1}
   (\bz_2 -\bmu_2)
\end{align*}
where the last equality follows from 
$ - (\bSigma_{11}^*)^{-1}\bSigma_{12}^*
 = \bSigma_{12}\bSigma_{22}^{-1}$ iff 
$-\bSigma_{12}^*\bSigma_{22} = 
\bSigma_{11}^*\bSigma_{12}$ iff 
$[\bSigma_{11}^* \; \bSigma_{12}^*] 
[\bSigma_{12} \bSigma_{22}]'\, = 0$ 
which follows from $\bSigma^{-1}\bSigma = I$. 
\end{proof}

\section{Proof of Proposition 3}
\label{A:chisq}
\begin{proof}
We will use the notation  
$\bhtheta_1 = \argmin_{\btheta_1} 
 ||\bSigma^{-1/2}(\bz - \bmu)- 
 (\bSigma^{-1/2})_{\bdot 1} \btheta_1||_2^2$ for
the OLS fit, and 
$\mbox{RSS}_k  = ||\bSigma^{-1/2}(\bz - \bmu)- 
 (\bSigma^{-1/2})_{\bdot 1} \bhtheta_1||_2^2$.

\noindent \textbf{Part 1.}\\
For $k=d$ we know from \eqref{eq:MDreg} that
$\mbox{RSS}_d = 0$. 
Now let $1 \ls k \ls d-1$. We want to show that 
$\mbox{RSS}_k = (\bz_2 - \bmu_2)'\,
 \bSigma_{22}^{-1} (\bz_2-\bmu_2)$. 
Let $\bhtheta \coloneqq 
[\bhtheta_1'\, \; 0 \; \ldots \; 0]'$
be the $d$-variate vector with coefficients 
$\bhtheta_1$ followed by $d - k$ zeroes.

We now have that 
\begin{align*}
\mbox{RSS}_k  &= ||\bSigma^{-1/2}(\bz-\bmu)- 
  (\bSigma^{-1/2})_{\bdot 1} \bhtheta_1||_2^2\\
&= ||\bSigma^{-1/2}(\bz - \bmu)- 
  \bSigma^{-1/2}\bhtheta||_2^2\\
&= (\bz - \bmu - \bhtheta)'\, \bSigma^{-1}
   (\bz - \bmu - \bhtheta)\\
&= [\bz_1 - \bmu_1 - \bhtheta_1 \;\;\;
    \bz_2 - \bmu_2]'\,\bSigma^{-1}\,
	 [\bz_1 - \bmu_1 - \bhtheta_1 \;\;\; 
	  \bz_2 - \bmu_2]\;\;.
\end{align*}

Following page 47 of \cite{Petersen2012} we can 
write $\bSigma^{-1} = \bA \bB \bA'\,$ with 
$$\bA \coloneqq 
\begin{bmatrix}
\bI & \bzero  \\
-\bSigma_{22}^{-1}\bSigma_{21} & \bI
\end{bmatrix}
\quad \quad \mbox{ and } \quad \quad 
\bB \coloneqq 
\begin{bmatrix}
\bC_1^{-1} & \bzero  \\
\bzero & \bSigma_{22}^{-1}
\end{bmatrix}$$
where $\bC_1 \coloneqq \bSigma_{11} - \bSigma_{12}
 \bSigma_{22}^{-1}\bSigma_{21}$\,.
We now have that
\begin{align*}
  [\bz_1 - \bmu_1 - \bhtheta_1 \;\;\; 
	 \bz_2 - \bmu_2]'\, \bA
	&= [\bz_1 - \bmu_1 - \bhtheta_1 -
	 (\bz_2 - \bmu_2)\bSigma_{22}^{-1}\bSigma_{21} 
	\;\;\; \bz_2 - \bmu_2]' \\
	&= [\bzero \;\;\; \bz_2 - \bmu_2]'
\end{align*}
using the result of Proposition 2. Therefore,
\begin{align*}
\mbox{RSS}_k  &= [\bz_1 - \bmu_1 - \bhtheta_1
    \;\;\; \bz_2 - \bmu_2]'\,\bSigma^{-1}\,
	 [\bz_1 - \bmu_1 - \bhtheta_1 
	  \;\;\; \bz_2 - \bmu_2]\\
&= [\bzero \;\;\; \bz_2 - \bmu_2]'\,\bB\,
   [\bzero \;\;\; \bz_2 - \bmu_2]\\
&= (\bz_2 - \bmu_2)'\,\bSigma_{22}^{-1}\,
   (\bz_2 - \bmu_2) \;\;.
\end{align*}

\noindent \textbf{Part 2.}\\
We will now show that the differences in $\mbox{RSS}$ 
follow a $\chi^2(1)$ distribution, that is\linebreak
$\Delta_k \coloneqq \mbox{RSS}_{k-1} -
 \mbox{RSS}_k \sim \chi^2(1)$ assuming
that $\bz = [\bz_1' \; \bz_2']'$ is multivariate
Gaussian with mean $\bmu$ and covariance matrix
$\bSigma$.
For $k=0$ we set by convention 
$\bhtheta \coloneqq \bzero$ 
and $\mbox{RSS}_0 \coloneqq
 ||\bSigma^{-1/2}(\bz - \bmu)||_2^2$.
 
We show the result for $k = 1$ as the subsequent steps
are analogous. 
The reasoning below is similar to Appendix A.2 of
\cite{danilov2010} where the cells were not yet 
ranked from most to least outlying.
As in Part 1 of the proof we can write
$\bSigma^{-1} = \bA \bB \bA'\,$ with 
$$\bA = 
\begin{bmatrix}
 1 & \bzero  \\
 -\bSigma_{22}^{-1}\bSigma_{21} & \bI
\end{bmatrix}
\quad \quad \mbox{ and } \quad \quad 
\bB = 
\begin{bmatrix}
C_1^{-1} & \bzero  \\
\bzero & \bSigma_{22}^{-1}
\end{bmatrix}$$
where this time $C_1 = \bSigma_{11} - \bSigma_{12}
 \bSigma_{22}^{-1}\bSigma_{21}$ is a scalar.
We can then write
\begin{align*}
\mbox{RSS}_0 
&= (\bz-\bmu)'\,\bSigma^{-1}(\bz-\bmu)\\
&= (\bz-\bmu)'\,\bA \bB \bA'\,(\bz-\bmu)\\
&= [z_1 - \mu_1 - (\bz_2 - \bmu_2)
    \bSigma_{22}^{-1}\bSigma_{21} \;\;\; 
		\bz_2 - \bmu_2]'\,\bB\, 
   [z_1 - \mu_1 - (\bz_2 - \bmu_2)
    \bSigma_{22}^{-1}\bSigma_{21} \;\;\; 
		\bz_2 - \bmu_2]	\\	
&= \left((z_1 -\mu_1^*)/\sigma_1^*\right)^2 +
   (\bz_2 - \bmu_2)'\,\bSigma_{22}^{-1}(\bz_2 - \bmu_2)\\
&= ((z_1 -\mu_1^*)/\sigma_1^*)^2 + \mbox{RSS}_1 
\end{align*}
where $\mu_1^* \coloneqq \mu_1 + (\bz_2 - \bmu_2)
\bSigma_{22}^{-1}\bSigma_{21}$ and 
$\sigma_1^* \coloneqq \sqrt{C_1}$. 
So we obtain 
$$\Delta_1 = \mbox{RSS}_0 - \mbox{RSS}_1 =
  ((z_1 -\mu_1^*)/\sigma_1^*)^2$$
and this is the square of a standard Gaussian 
variable since $z_1 - \mu_1^*$
is Gaussian with expectation 0 and standard 
deviation $\sigma_1^*$\,. 
We thus have $\Delta_1 \sim \chi^2(1)$.
\end{proof}

\section{Implementation of the cellHandler algorithm}
\label{A:cellHandler}
The LAR component of cellHandler is a regression of
$\btY$ on $\bsX$ as defined in the paper. 
Since this regression has no intercept and we need
to preserve the column scaling in $\bsX$, we run 
the function {\it lars::lar} with the options 
{\it intercept=F} and {\it normalize=F}.

For the imputations in Proposition 2 and the
RSS in Proposition 3 we require the OLS fits
$\,\bhtheta_A\,$ minimizing
$\;||\bSigma^{-1/2}(\bz-\bmu)-
   (\bSigma^{-1/2})_A\, \btheta_1||_2^2\;$ 
where $A$ is the set of active predictor variables
in every step of LAR. 
Fortunately, these can be obtained without 
significant additional computation time because each 
step of LAR already carries out the QR decomposition
of $\,(\bsX_A)'\,\bsX_A\,$  where $\bsX_A$ is the 
submatrix of $\bsX$ consisting of the columns in $A$.
The resulting OLS regression vectors $\,\bhbeta_A\,$ 
obtained by LAR (which contain zeroes for the 
inactive variables) are then easily rescaled 
to $\,\bhtheta_A = \bW^{-1} \bhbeta_A$\;.

\section{Description of the initial estimator DDCW}
\label{A:DDCW}

The Detection-Imputation (DI) method of Section
\ref{sec:DIalgo} needs initial cellwise robust
estimates $\bhmu^0$ and $\bhSigma^0$ of 
location and covariance.
One option is to insert the 2SGS estimator of
\cite{Leung2017}.
We also developed a different initial estimator
called DDCW, which we describe here.
Its steps are:
\begin{enumerate}
\item Drop variables with too many missing values
      or zero median absolute deviation, and
			continue with the remaining columns.
\item Run the DetectDeviatingCells (DDC) method
      \citep{Rousseeuw2018} with the constraint
			that no more than $n\,maxCol$ cells are
			flagged in any variable.
			DDC also rescales the variables, and may
			delete some cases.
			Continue with the remaining imputed and
			rescaled cases denoted as $\bz_i$\,.
\item Project the $\bz_i$ on the axes of
      their principal components, yielding the 
			transformed data points	$\btz_i$\;.
\item Compute the wrapped location $\bhmu_w$
      and covariance matrix $\bhSigma_w$
			\citep{Raymaekers2018} of these 
			$\btz_i$\,. 
			Next,	compute the temporary points
			$\bu_i = (u_{i1},...,u_{id})$
			given by $u_{ij} = \max\{\min\{
			\tilde{z}_{ij}-(\bhmu_w)_j,2\},-2\}$.
			Then remove all cases for which the 
			squared robust distance
			$\RD^2(i) = 
			 \bu_i' \bhSigma_w^{-1} \bu_i$
			exceeds 
			$\chi^2_{d,q} 
			 \median_h(\RD^2(h))/\chi^2_{d,0.5}$\;.
\item Project the remaining $\btz_i$ on the
      eigenvectors of $\bhSigma_w$ and
			again compute a wrapped location and
			covariance matrix.
\item Transform these estimates back to
      the original coordinate system of
			the imputed data, and undo the
			scaling.
			This yields the estimates 
			$\bhmu^0$ and $\bhSigma^0$\,.
\end{enumerate}
Note that DDCW can handle missing values since
the DDC method in step 2 imputes them.
The reason for the truncation in the rejection
rule in step 4 is that otherwise the robust
distance $\RD$ could be inflated by a single
outlying cell.
Step 4 tends to remove rows which deviate
strongly from the covariance structure. 
These are typically rows which cannot be shifted 
towards the majority of the data without 
changing a large number of cells.

\section{Step by step 
  description of the DI algorithm}
\label{A:DI}
We now give a step-by-step description of the
DI algorithm, with some additional details.
\begin{enumerate}
\item Standardize the columns (variables) as 
      described in the beginning of 
	    Section \ref{sec:ranking}.
\item Compute initial estimates $\bhmu^0$ 
      and $\bhSigma^0$. The algorithm 
			currently has two options for this:
	\begin{itemize}
	\item the DDCW estimator described in Section
        \ref{A:DDCW} above;
	\item the 2SGS estimator of \cite{Leung2017},
	      available in the \texttt{R} package
				GSE \citep{Leung2019}.
	\end{itemize}
\item {\bf D-step}.
  Given the estimates $\bhmu^{t-1}$ and 
	$\bhSigma^{t-1}$ where $t=1,2,\ldots$ we 
	flag outlying cells across all rows of 
	the dataset. This is done as described in
	Section \ref{sec:DIalgo} by applying the
	cellHandler method of Section 
	\ref{sec:cellHandler} to each row $\bz_i'$ 
	using $\bhmu^{t-1}$ and $\bhSigma^{t-1}$.
	The D-step imposes a maximum on the number of
  flagged cells in a row, namely 
  $n\, \mbox{\it maxCol}$ where {\it maxCol} 
  is set to 25\% by default.
  Since all missing values (NA's) are automatically 
  flagged, the algorithm would not be able to run 
  if there are too many NA's in a column.
  In practice, the algorithm starts by setting
  variables with too many NA's aside and giving a
  message about this. The D-step yields a list of 
	flagged cells in each row, which contains 
	the flagged outlying cells as well as their 
	imputed values.
\item{\bf I-step}. We re-estimate the center as 
  $\bhmu^t$ which is the mean of the dataset with 
	its imputed cells.
  For computing $\bhSigma^t$ we use the formula of 
	the M-step in the EM-algorithm. 
	It is not just the covariance matrix of the 
	imputed data, since this would underestimate the 
	true variability.
	Therefore the EM method adds a bias correction.
  This bias correction depends on which cells were 
  imputed, and can therefore be different for every 
  row of the data. 
  Suppose the first row $\bz_1$ has an imputed part 
  $\bz_{1i}$ and an untouched part $\bz_{1u}$\,, 
  then the bias correction matrix from that row is
  $$B_{ii} = \frac{1}{n}\bhSigma^{t-1}_{uu} -
    \frac{1}{n}\bhSigma^{t-1}_{iu} 
	  (\bSigma^{t-1}_{uu})^{-1} 
	  \bSigma^{t-1}_{ui}\;.$$ 
  This correction term is known to remove the bias 
	when the data is uncontaminated multivariate Gaussian
  with missing values generated completely at random
  (MCAR), that is, independent of both the observed
  cells as well as the values the missing cells
  had before they became unavailable. 
  Also in our simulations with contaminated data this 
  bias correction turned out to improve the results.
\item Iterate steps 3 and 4 alternatingly until
	$$||\bhmu^{t} - \bhmu^{t-1}||_2^2 +
	  ||\bhSigma^{t} - \bhSigma^{t-1}||_2^2$$ 
	is below a given tolerance, where the 
  second norm is given by \eqref{eq:Frobenius}.
\item Apply cellHandler with the converged $\bhmu$ 
  and $\bhSigma$ to obtain the final list of cellwise 
	outliers and their imputed values.
\item Unstandardize the results using the univariate 
  location and 
	scale estimates of the original data columns,
	used in step 1.
\end{enumerate}

\section{Proof of Proposition 4}
\label{A:KL}
\begin{proof} Starting from the well-known
formula for $\KL(\bX,\bY)$ we obtain
\begin{align*}
\KL(\bX,\bY) 
&= \tr(\bA\,\bB^{-1}) - d 
   - \log\det(\bA\,\bB^{-1})
   \nonumber \\
&= \tr(\bA\,\bB^{-1/2}\bB^{-1/2}) - d - 
   \log \det(\bA\,\bB^{-1/2}\bB^{-1/2})
	 \nonumber \\
&= \tr(\bB^{-1/2}\bA\,\bB^{-1/2}) - d - 
   \log \det(\bB^{-1/2}\bA\,\bB^{-1/2})
	 \nonumber \\
&= \Big( \sum_{j=1}^d \eta_j \Big) - d 
   - \log \Big( \prod_{j=1}^d \eta_j \Big)
	 \nonumber \\
&= \sum_{j=1}^d (\eta_j - 1 
   - \log(\eta_j)) = D(\bA,\bB) 	
\end{align*}
where the fourth equality used the fact that
$\bB^{-1/2}\bA\,\bB^{-1/2}$ is PSD so it can be
diagonalized, hence its trace is the sum of its
eigenvalues.
\end{proof}

\section{F-scores in dimensions 10, 20 and 40}
\label{A:Fscores}
Figure \ref{fig:cellHandler_perout0.2} showed
the precision, recall, and F-score for data
generated by the contaminated ALYZ model and 
the contaminated  A09 model, for
$n=400$ points in $d=20$ dimensions. 
Here Figure \ref{fig:Fscores} shows the F-scores 
for both DDCW.DI (DI starting from DDCW)
and 2SGS.DI (DI starting from 2SGS).
These are byproducts of the simulations in
Figure \ref{fig:sim_d10} 
for $(n,d)=(100,10)$ and 
Figure \ref{fig:KLdiv_highdim} 
for $(n,d)=(400,20)$ and $(800,40)$. 
Note that step 6 of the DI algorithm in
Appendix \ref{A:DI} provides the flagged
cells, so the DDCW.DI curves for $d=20$
in Figure \ref{fig:Fscores} correspond to 
those of cellHandler in the lower panel of 
Figure \ref{fig:cellHandler_perout0.2}.

\begin{figure}[!ht]
\center
\vskip-0.5cm
\includegraphics[width = 0.43\textwidth]
 {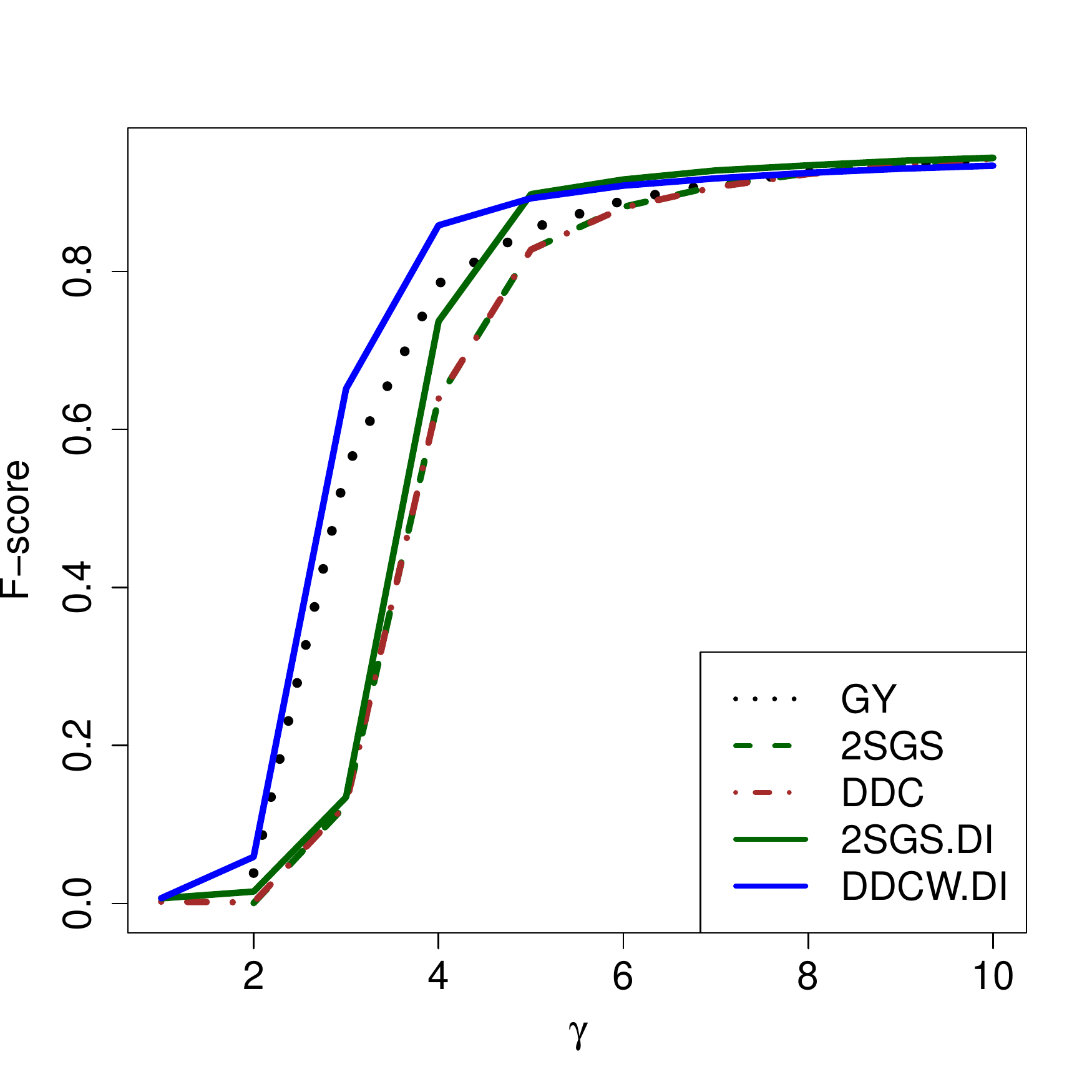}
\includegraphics[width = 0.43\textwidth]
 {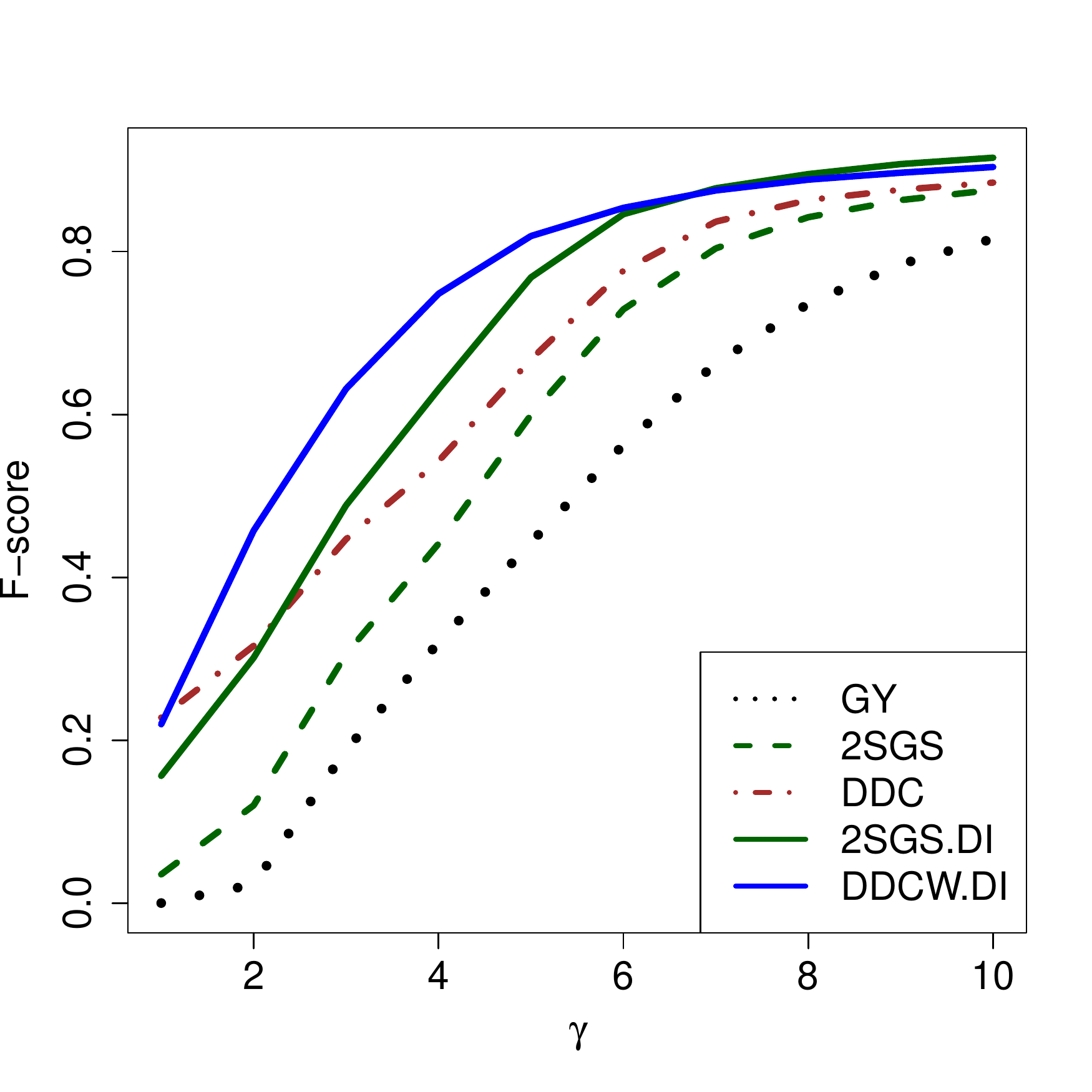}
\vskip-0.5cm
\includegraphics[width = 0.43\textwidth]
 {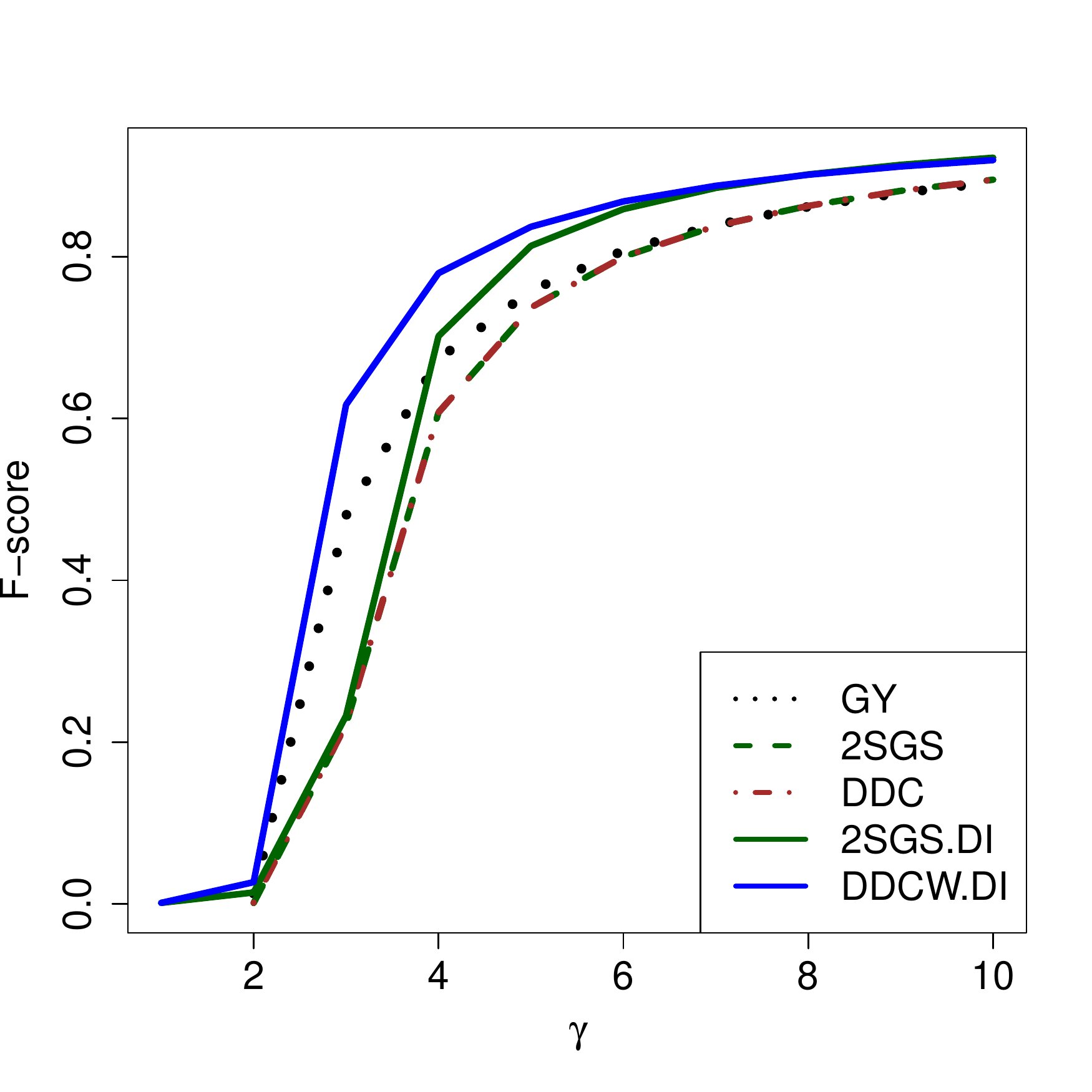}
\includegraphics[width = 0.43\textwidth]
 {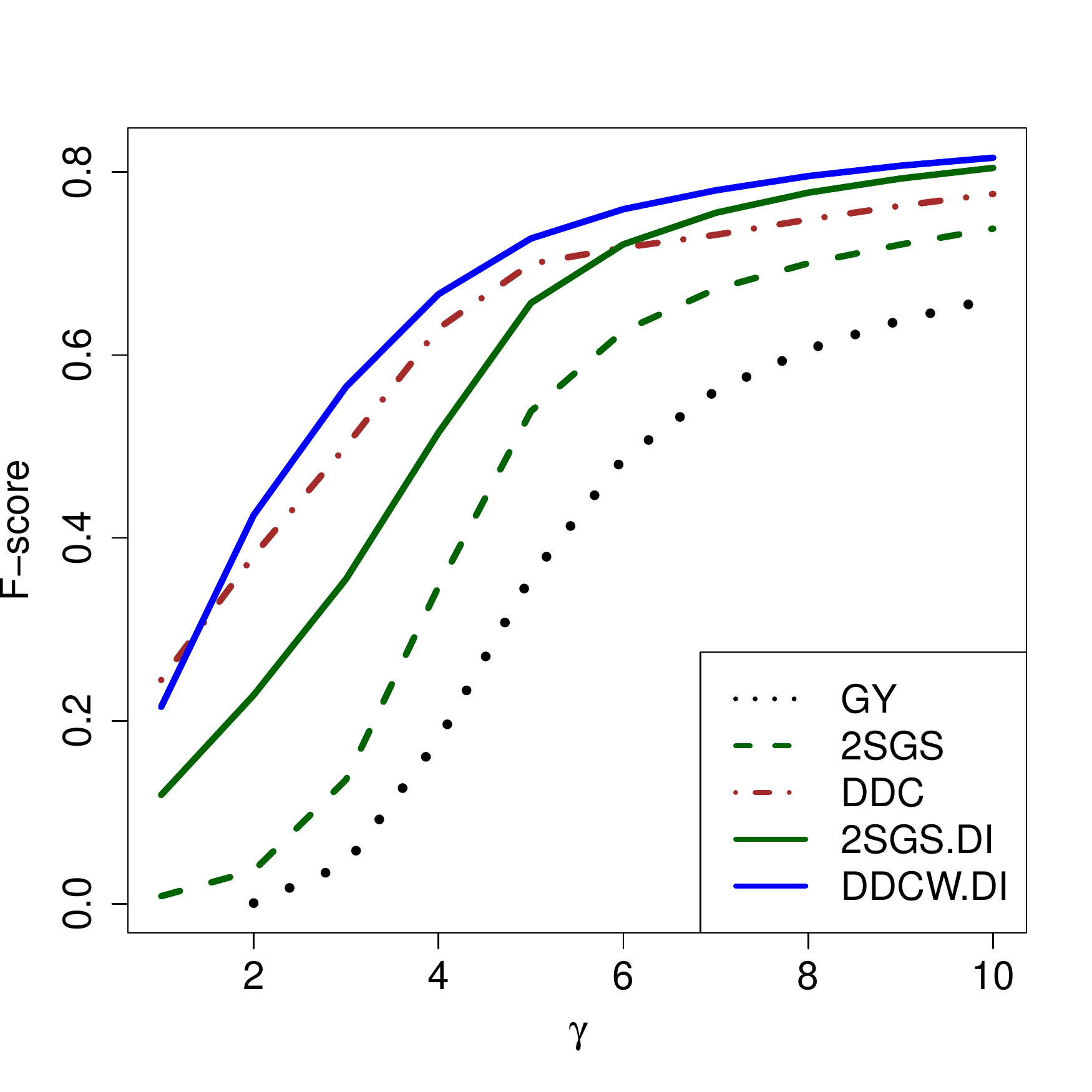}
\vskip-0.5cm
\includegraphics[width = 0.43\textwidth]
 {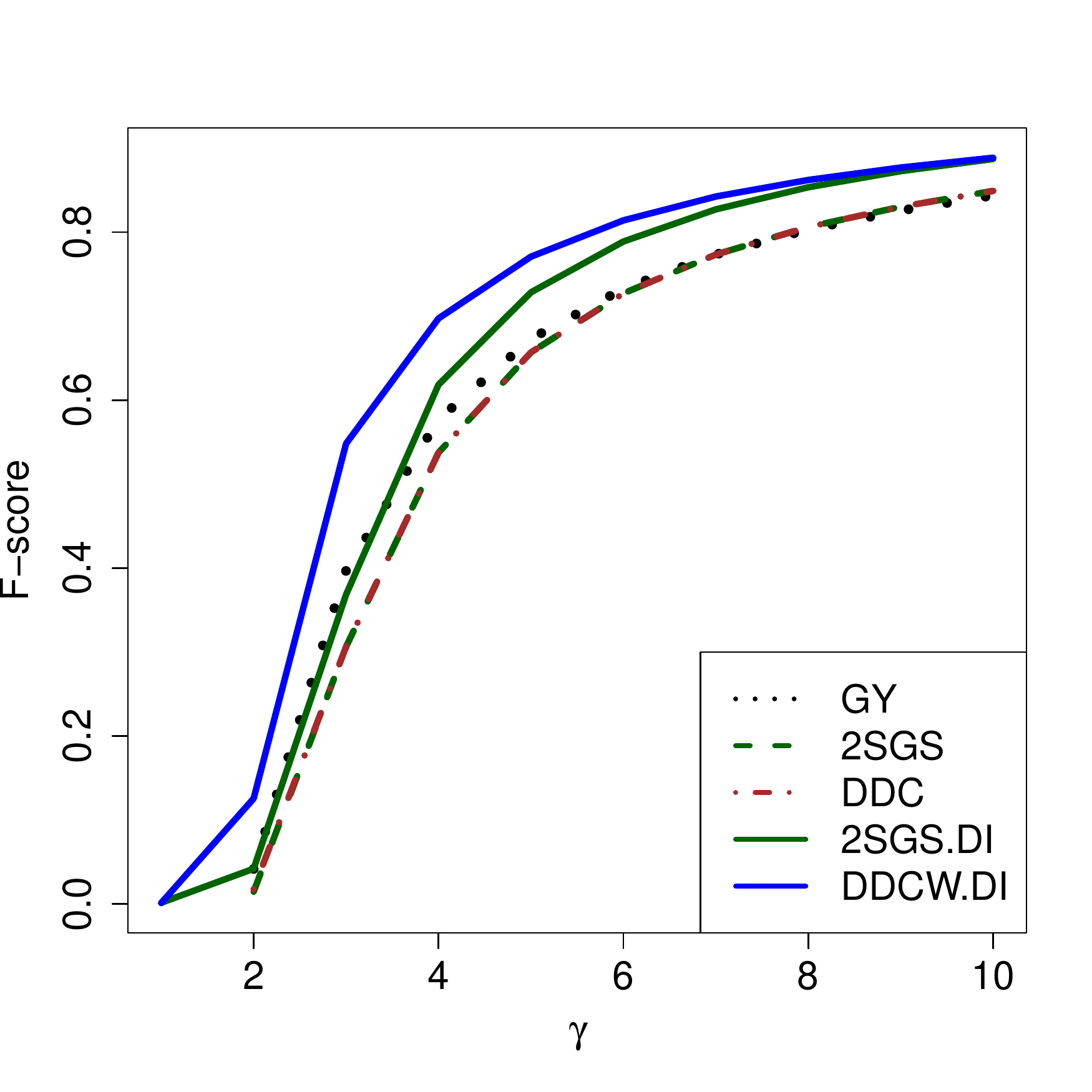}
\includegraphics[width = 0.43\textwidth]
 {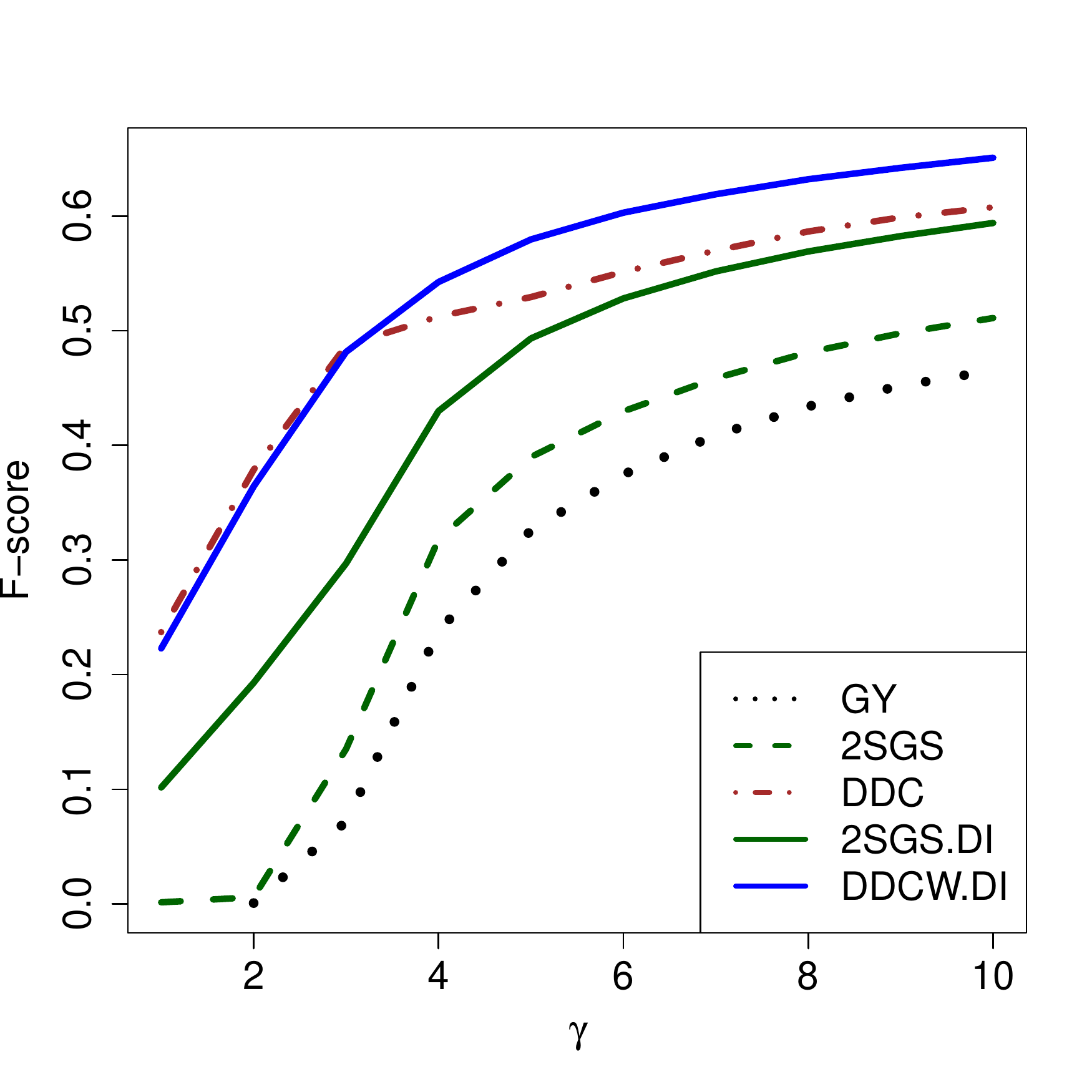}
\vskip-0.2cm
\caption{F-scores of flagged cellwise outliers 
   on data generated by 
	 the contaminated ALYZ model (left) and 
	 the contaminated  A09 model (right), for
	 $(n,d)=(100,10)$ (top),
	 $(n,d)=(400,20)$ (middle), and
	 $(n,d)=(800,40)$ (bottom).}
\label{fig:Fscores}
\end{figure}

\section{Computation times of DI}
\label{A:comptimes}
Table \ref{tab:times} shows computation 
times of the DI algorithm as implemented in the  
\texttt{R} package \texttt{cellWise} on CRAN.
This implementation contains compiled C++ code.
The times are the averaged runtimes of one 
replication of the simulations in
Figure \ref{fig:sim_d10} 
for $(n,d)=(100,10)$ and 
Figure \ref{fig:KLdiv_highdim} 
for $(n,d)=(400,20)$ and $(800,40)$. 
The times are in seconds on a laptop with
Intel Core i7-5600U at 2.60Hz.
We see that DDCW.DI requires less time than 
2SGS.DI because the DDCW initial estimator
is quite fast. Note that the computation
time of the D-step in DI could be reduced by 
executing the LARS computations on the rows 
in parallel.

\begin{table}[ht]
\centering
\caption{Computation time (in seconds) of 
one replication of DI in the simulation.} 
\label{tab:times}
\begin{tabular}{crr}
\hline
 $n$ and $d$ & 2SGS.DI & DDCW.DI \\ 
\hline
 $n=100$, $d=10$ & 1.70 & 0.54 \\ 
 $n=400$, $d=20$ & 16.42 & 5.38 \\ 
 $n=800$, $d=40$ & 110.25 & 35.16 \\ 
\hline
\end{tabular}
\end{table}

\FloatBarrier
\section{Simulation with cellwise and
            casewise outliers}
\label{A:sim}
We now run a simulation in which the
data are contaminated by $10\%$ of cellwise 
outliers generated as in the paper, plus $10\%$ of 
rowwise outliers.
In this particular setting ``rowwise outliers'' 
refers to rows in which all cells are contaminated 
in the same way as before, that is, rows with $d$ 
cellwise outliers.
We generate these outlying rows by the formula
$\bv = \gamma d\sqrt{d}\, \bu'/\MD(\bu,\bmu, 
 \bSigma)$
where $\bu$ is the eigenvector of $\bSigma$ with 
smallest eigenvalue.
This corresponds to the cellwise formula 
of Subsection \ref{subsec:cFsim} in which
the indices of the outlying cells 
$K = \{j(1),\ldots,j(k)\}$ are replaced by
$K = \{1,\ldots,d\}$.
Next, we replace 10\% of the rows by $\bv$, and
afterward sample the positions of the cellwise 
outliers from the remaining 90\% of the rows.
The results are shown in Figures \ref{fig:both_d10}
and \ref{fig:both_highdim}. 
They look qualitatively similar to those in
Figures \ref{fig:sim_d10} and \ref{fig:KLdiv_highdim}
in the paper. 

\begin{figure}[!ht]
\centering
\vskip1cm
\begin{minipage}{0.49\linewidth}
  \centering 
    \textbf{ALYZ, 10\% cells \& 10\% cases,\\
		        $\bm{d = 10}$}
	\includegraphics[width=0.9\textwidth]
	 {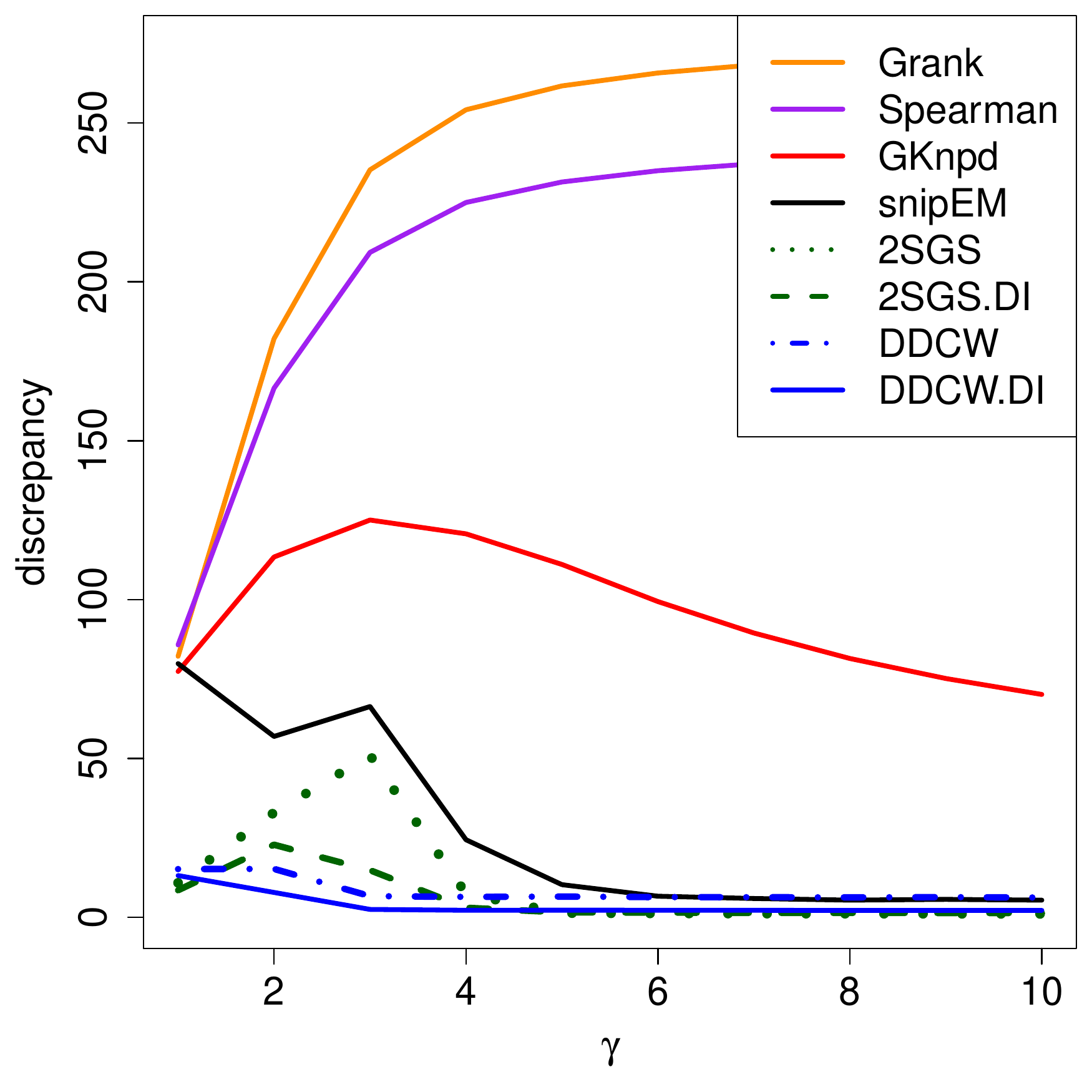} 
\end{minipage}
\begin{minipage}{0.49\linewidth}
  \centering
	  \textbf{A09, 10\% cells \& 10\% cases,\\ 
	         	$\bm{d = 10}$}
  \includegraphics[width=0.9\textwidth]
	 {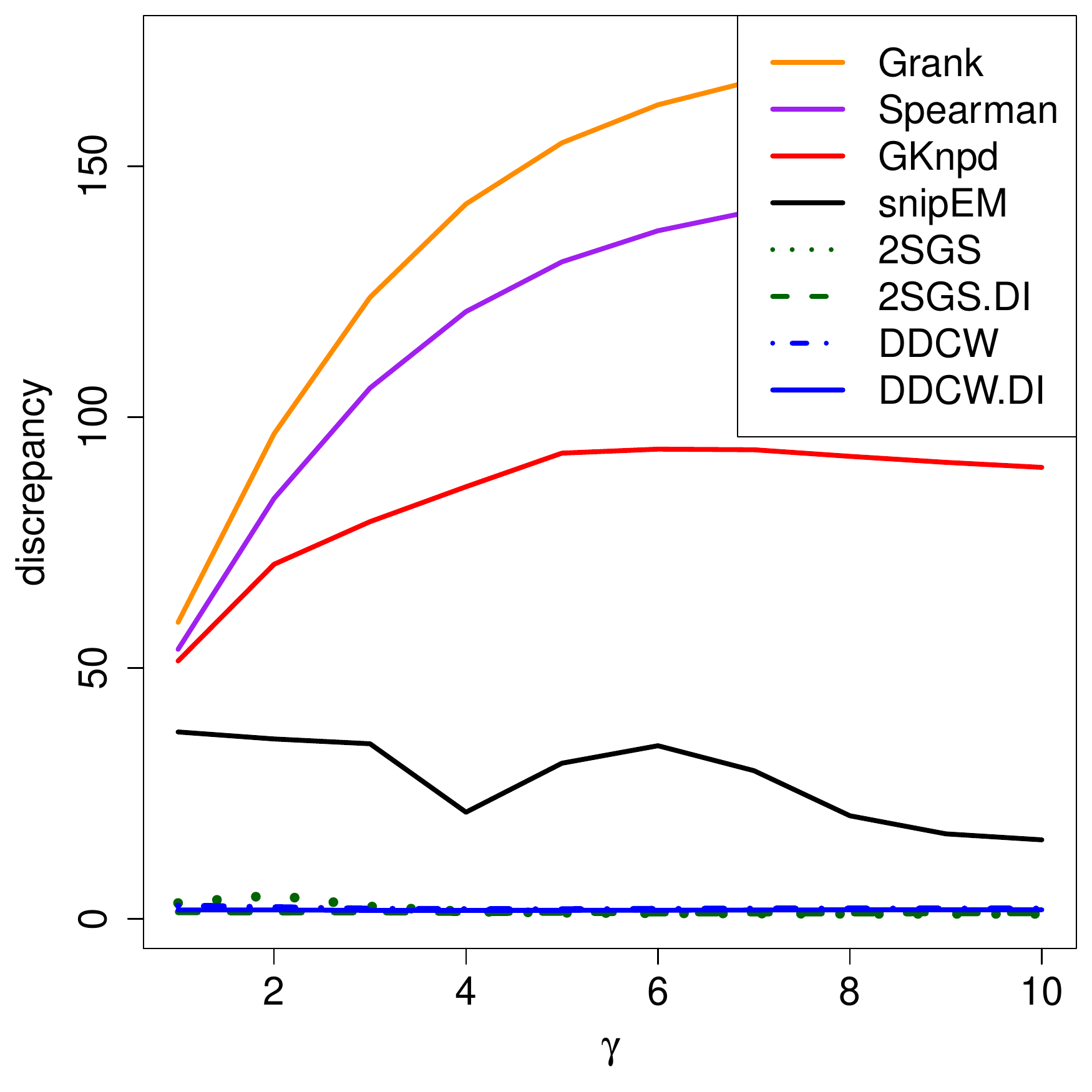} 
\end{minipage}
\vskip-0.3cm
\caption{Discrepancy 
  $D(\widehat{\bSigma},\bSigma)$
  of estimated covariance
  matrices for $d = 10$, $n = 100$.}
\label{fig:both_d10}
\end{figure}

\begin{figure}[!ht]
\centering
\vskip0.5cm
\begin{minipage}{0.49\linewidth}
  \centering 
    \textbf{ALYZ, 10\% cells \& 10\% cases,\\
		        $\bm{d = 20}$}
	\includegraphics[width=0.9\textwidth]
	 {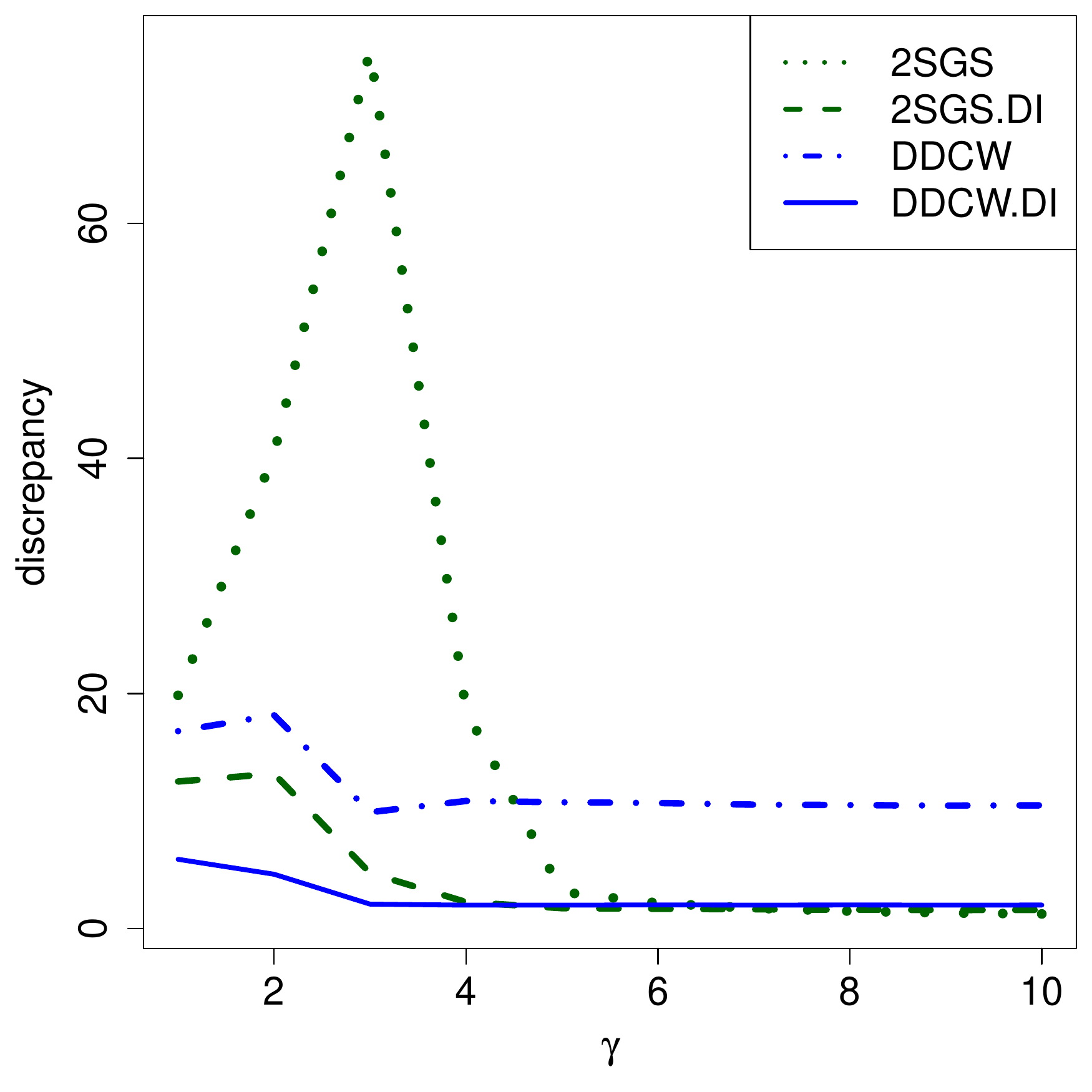} 
\end{minipage}
\begin{minipage}{0.49\linewidth}
  \centering
	  \textbf{A09, 10\% cells \& 10\% cases,\\ 
		        $\bm{d = 20}$}
  \includegraphics[width=0.9\textwidth]
	{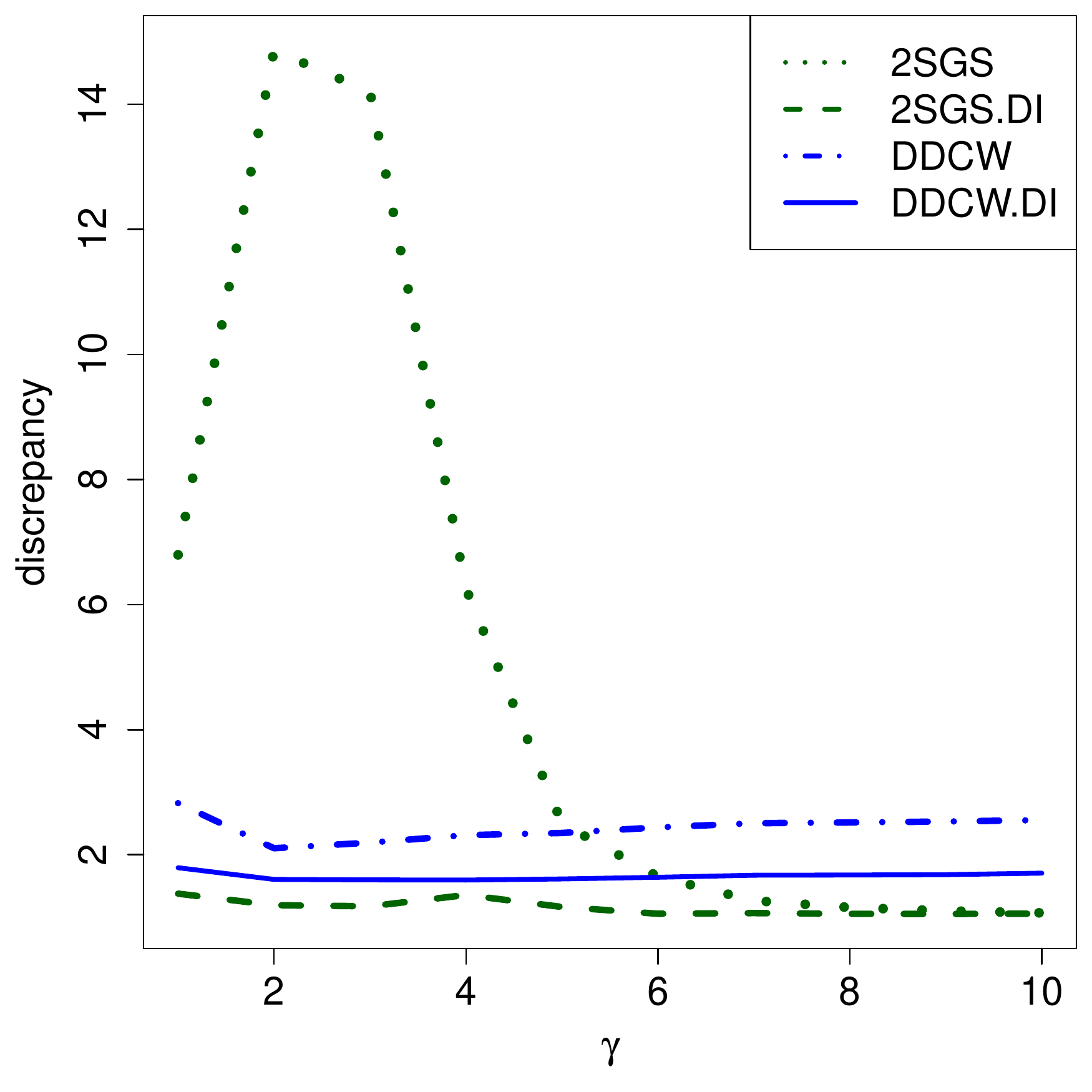} 
\end{minipage}
\vskip0.3cm
\begin{minipage}{0.49\linewidth}
  \centering 
    \textbf{ALYZ, 10\% cells \& 10\% cases,\\ 
		        $\bm{d = 40}$}
	\includegraphics[width=0.9\textwidth]
	 {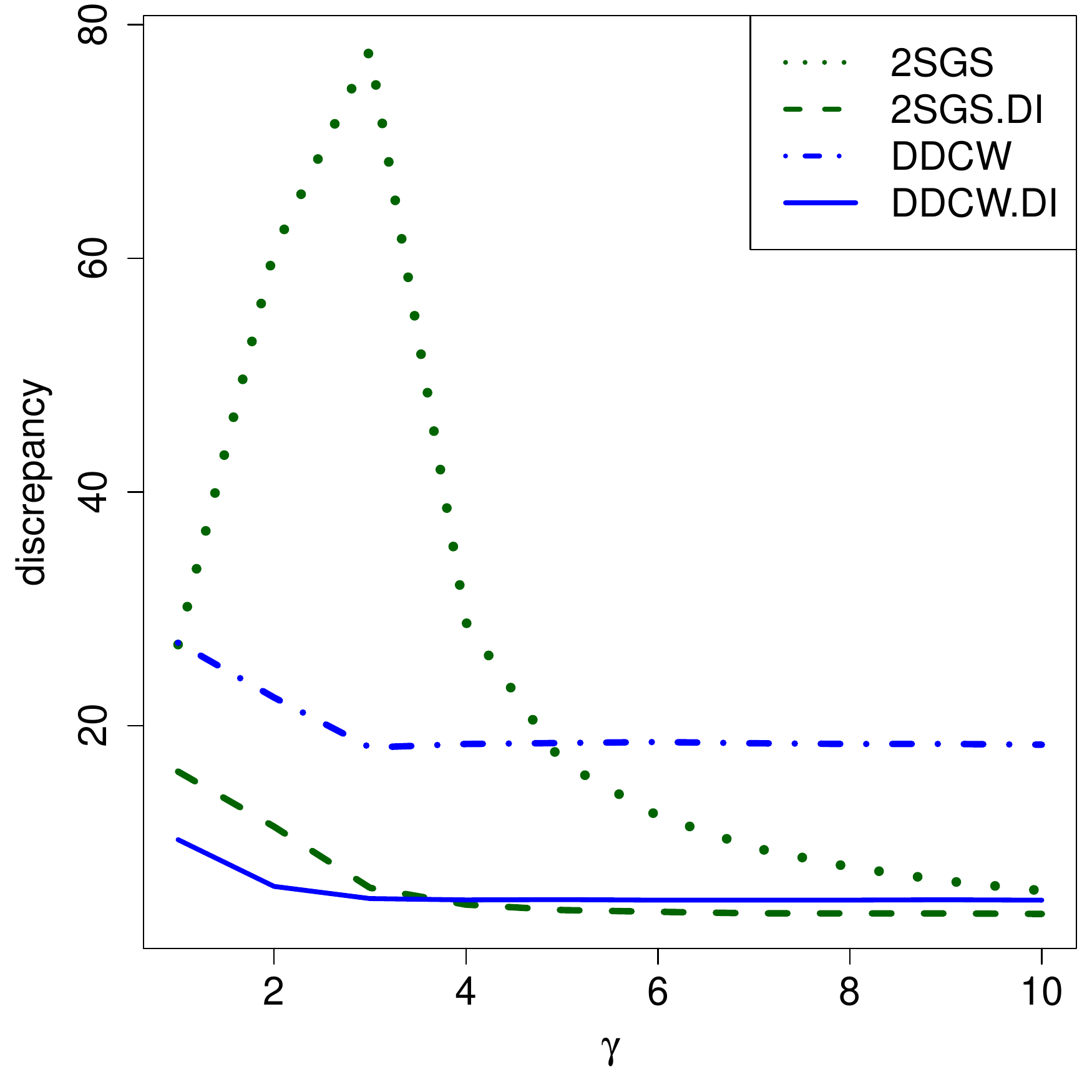} 
\end{minipage}
\begin{minipage}{0.49\linewidth}
  \centering
	  \textbf{A09, 10\% cells \& 10\% cases,\\ 
		        $\bm{d = 40}$}
  \includegraphics[width=0.9\textwidth]
	 {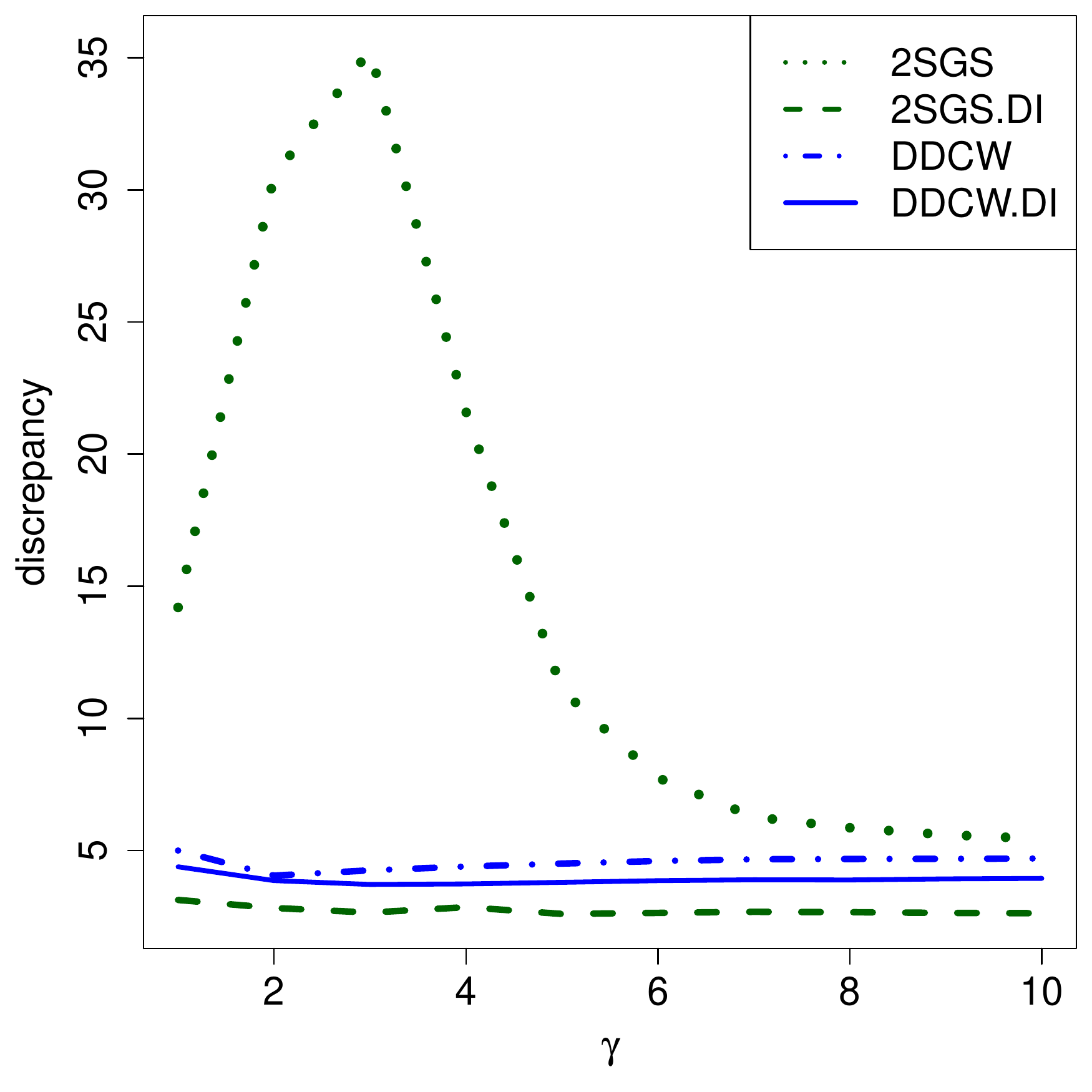} 
\end{minipage}
\vskip-0.3cm
\caption{Discrepancy $D(\widehat{\bSigma},\bSigma)$ 
  given by \eqref{eq:discr} of estimated covariance
  matrices for $d = 20$ and $n = 400$ (top panels) 
  and for $d = 40$ and $n = 800$ (bottom panels).}
\label{fig:both_highdim}
\end{figure}

Note that in this section we do not plot F-scores
for rowwise outliers, because the cellHandler and 
DI algorithms do not have a mechanism for detecting 
rowwise outliers. They are methods for cellwise 
outliers.
If we were to try to detect both types of outliers
simultaneously we would run into an identifiability 
issue. 
For instance, you can easily generate rows which 
would be considered outliers under the rowwise 
paradigm, but only contain a single cellwise 
outlier under the cellwise paradigm. 
And starting from the cellwise paradigm, how many 
cellwise outliers in a row would it take before 
the entire row should be considered outlying?
The identifiability issue is especially complicated
in the current simulations because we generate 
cellwise outliers in a structured way as explained
in Section 2.3, so that they do not stand out 
individually.

Because of these concerns, when 
generating both types of outliers in this section 
our focus is on the accuracy of the covariance 
matrix estimate, so it is clear what the goal is. 
Afterward, the user can choose whether to employ 
the estimated covariance matrix to detect rowwise 
outliers based on their Mahalanobis-type distance, 
or as an input to the cellHandler algorithm to 
detect cellwise outliers.\\

\FloatBarrier
\clearpage
\section{List of volatile organic compounds}
\label{A:VOCs}
The volatile organic compounds (VOC's)
analyzed in Section \ref{sec:example}
are listed below.

\begin{table}[ht!]
\begin{tabular}{ll}
\hline
Variable Name & VOC name \\
\hline
URX2MH & 2-Methylhippuric acid\\
URX34M & 3- and 4-Methylhippuric acid\\
URXAAM & N-Acetyl-S-(2-carbamoylethyl)-L-cysteine\\
URXAMC & N-Acetyl-S-(N-methylcarbamoyl)-L-cysteine\\
URXATC & 2-Aminothiazoline-4-carboxylic acid\\
URXBMA & N-Acetyl-S-(benzyl)-L-cysteine\\
URXCEM & N-Acetyl-S-(2-carboxyethyl)-L-cysteine\\
URXCYM & N-Acetyl-S-(2-cyanoethyl)-L-cysteine\\
URXDHB & N-Acetyl-S-(3,4-dihydroxybutyl)-L-cysteine\\
URXHP2 & N-Acetyl-S-(2-hydroxypropyl)-L-cysteine\\
URXHPM & N-Acetyl-S-(3-hydroxypropyl)-L-cysteine\\
URXIPM3 & N-Acetyl- S- (4- hydroxy- 2- methyl- 2- butenyl)-L-cysteine\\
URXMAD & Mandelic acid\\
URXMB3 & N-Acetyl-S-(4-hydroxy-2-butenyl)-L-cysteine\\
URXPHG & Phenylglyoxylic acid \\
URXPMM & N-Acetyl-S-(3-hydroxypropyl-1-methyl)-L-cysteine\\
\hline
\end{tabular}
\end{table}

\clearpage

\FloatBarrier
\newpage
\section{Analysis of the VOC data after preprocessing}
\label{A:VOCsclr}

As kindly pointed out by a referee, the 
concentrations of compounds in urine samples can 
depend on the dilution of the urine, if the 
analytical chemistry technique did not adjust
for this. 
Therefore, the measurements might not always be 
comparable across the different samples. 
This issue can be dealt with by applying the
centered log ratio (CLR) transform, which 
transforms the original data to log-ratios 
between the variables. 
We reran the analysis of the VOC data 
after applying the CLR transform using the 
\texttt{cenLR} function in the \texttt{R}-package 
\texttt{robCompositions} \citep{templ2011}.

\begin{figure}[!hb]
\center
\vskip-0.4cm
\includegraphics[width = 0.54\textwidth]
  {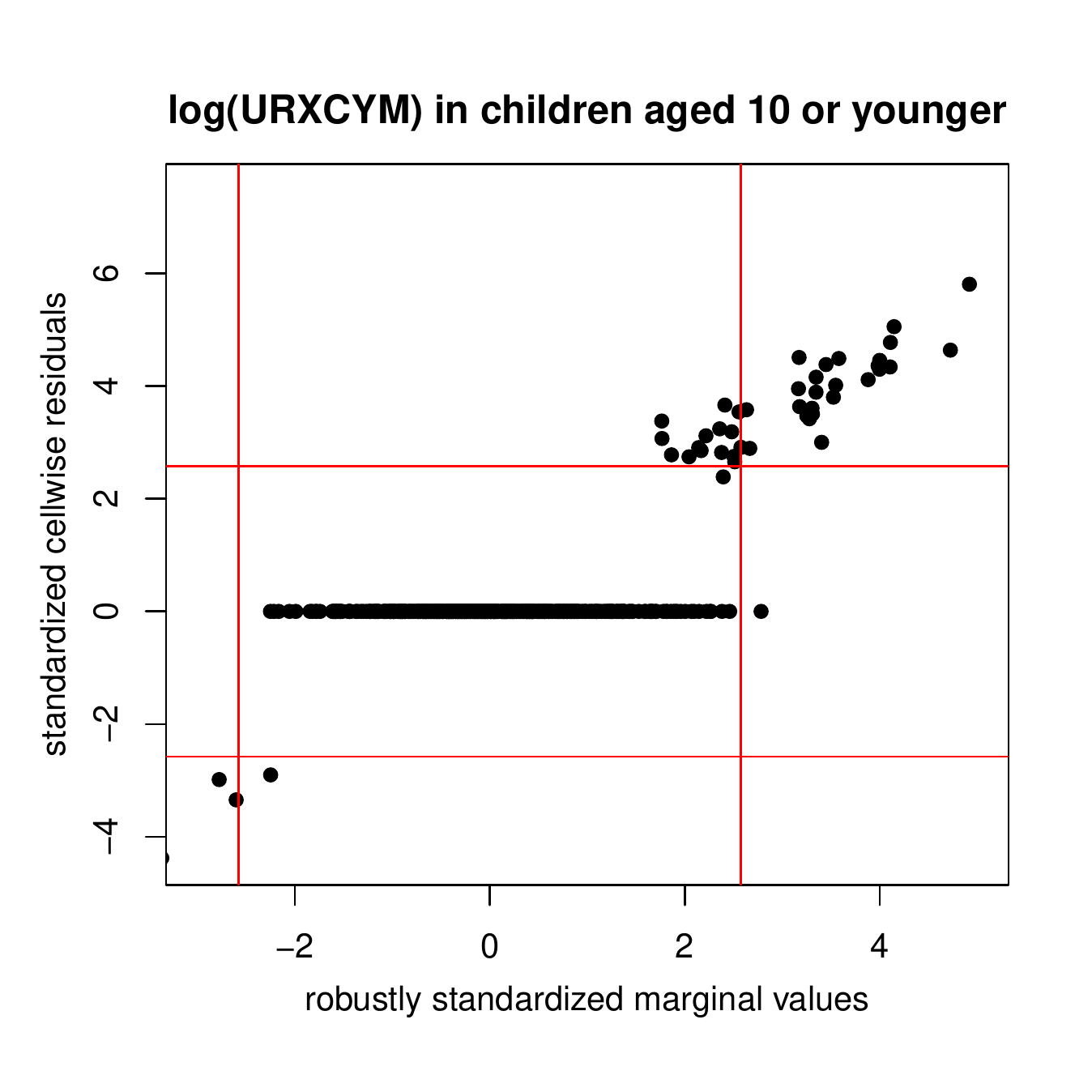}
\vskip-0.4cm
\caption{Plot of standardized cell residuals
of CLR-transformed URXCYM obtained by cellHandler, 
versus the robustly standardized values of
CLR-transformed URXCYM on its own.}
\label{fig:ZresVersusZ_clr}
\end{figure}

\FloatBarrier

Figures \ref{fig:ZresVersusZ_clr} and 
\ref{fig:VOCS_clr} present the results of this 
analysis. In the first plot we see that the 
univariate detection rule flags a few more 
outliers after the CLR transform than without it.
However, the red curve in the second figure still 
shows no effect of adult smokers in the household 
on URXCYM in children. Therefore the conclusion 
that univariate outlier detection is insufficient
here remains valid. 

\begin{figure}[!ht]
\center
\vskip0.2cm
\includegraphics[width = 0.58\textwidth]
  {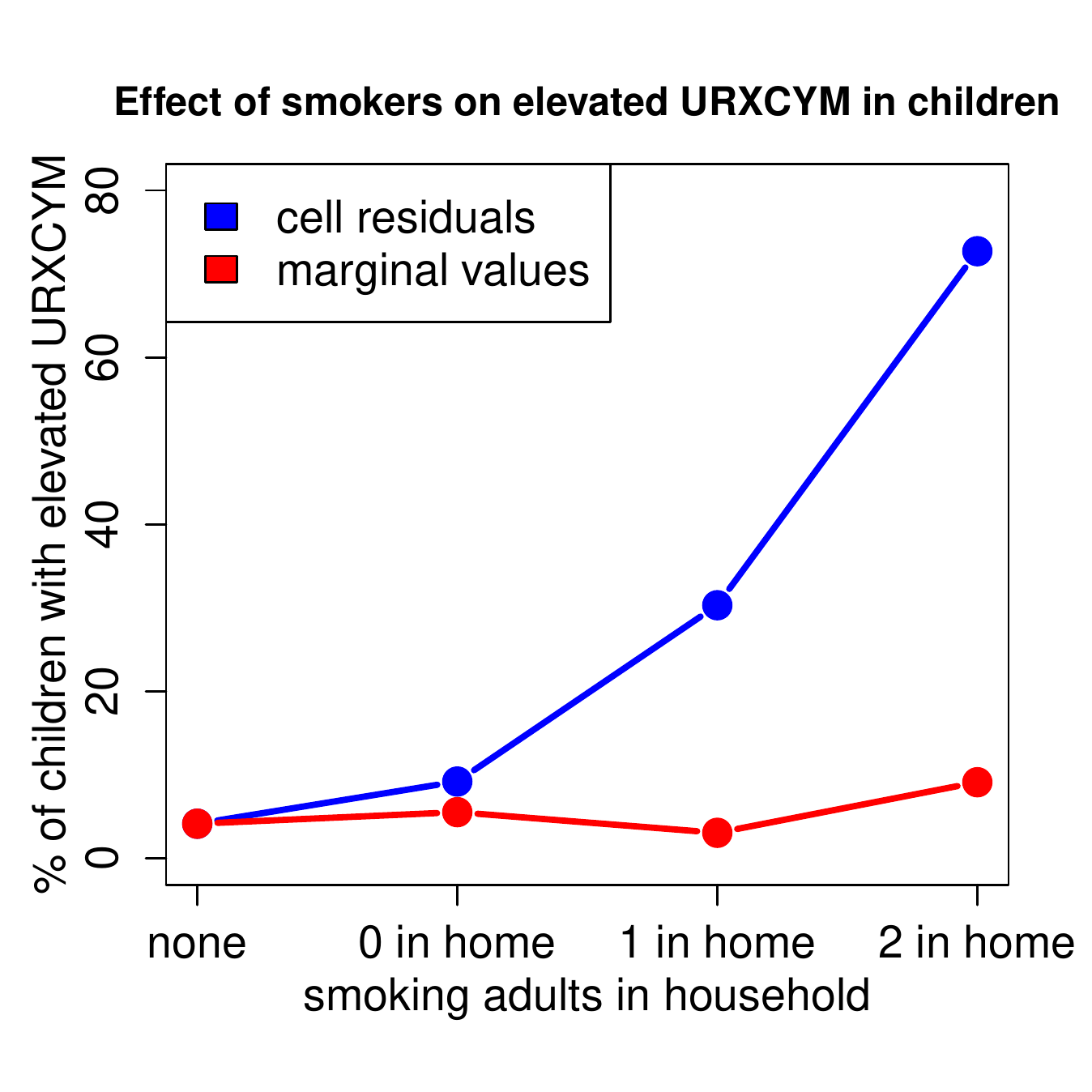}
\vskip-0.2cm
\caption{The blue curve shows the percentage 
of elevated URXCYM cell residuals in function 
of the smoking status of adult family members. 
The red curve shows the percentage of elevated
marginal URXCYM values. Here the URXCYM values 
were preprocessed with the CLR transformation.}
\label{fig:VOCS_clr}
\end{figure}

\end{appendix}

\end{document}